\documentclass[10pt,reqno]{amsart}

\usepackage{latexsym}
\usepackage{amssymb}
\usepackage{amsmath}
\usepackage{amsfonts}
\usepackage{amssymb} 

\numberwithin{equation}{section}

\newcommand{\cDD}{{\mathcal D}}
\newcommand{\cHH}{{\mathcal H}}

\newcommand{\de}{\partial}

\newcommand{\nat}{{\Delta_{\theta}}}
\newcommand{\nar}{{\Delta_{r}}}
\newcommand{\snat}{\sqrt{\Delta_{\theta}}}
\newcommand{\snar}{\sqrt{\Delta_{r}}}
\newcommand{\diag}{{\rm diag}}

\newcommand{\eqn}{\begin{eqnarray}}
\newcommand{\feqn}{\end{eqnarray}}

\newtheorem{theorem}{Theorem}
\newtheorem{lemma}{Lemma}
\newtheorem{corollary}{Corollary}
\newtheorem{remark}{Remark}

\begin{document}

\title[The Dirac Equation in Kerr-Newman-AdS Black Hole Background]{The Dirac Equation in Kerr-Newman-AdS Black Hole Background}

\author{Francesco Belgiorno}
\address{Dipartimento di Fisica,
Universit\`a degli Studi di Milano,  Via Celoria 16, 20133 Milano, Italy}
\email{belgiorno@mi.infn.it}
\author{Sergio L. Cacciatori}
\address{
Dipartimento di Fisica, Universit\`a degli Studi dell'Insubria,
Via Valleggio 11, 22100 Como, Italy}
\email{sergio.cacciatori@uninsubria.it}

\begin{abstract}
We consider the Dirac equation on the Kerr-Newman-AdS black hole background.
We first perform the variable separation for the Dirac equation and define the Hamiltonian operator 
$\hat H$. Then we show that for a 
massive Dirac field with mass $\mu\geq \frac{1}{2l}$ essential selfadjointness of $\hat H$ 
on $C_0^{\infty} ((r_+,\infty) \times S^2)^4$ is obtained 
even in presence of the boundary-like behavior of infinity in an asymptotically AdS black hole background. 
Furthermore qualitative spectral properties of the Hamiltonian are taken into account 
and in agreement with the existing results concerning the case of stationary 
axi-symmetric asymptotically flat black holes we infer the absence of time-periodic
and normalizable solutions of the Dirac equation around the black hole in the non-extremal case.
\end{abstract}

\maketitle

\section{Introduction}
Black holes play a very important role in many aspects of theoretical physics, starting from the fact that they 
provide nontrivial exact solutions of Einstein equations, going through their thermodynamical properties and 
ending up with their role of ``hydrogen atoms of quantum gravity''
\cite{Callan:1996ey}. Their relevance is fundamental also beyond the theoretical aspects.
The uniqueness and the no hair theorems (\cite{Chrusciel:1994sn})
give strong restrictions on the possible signatures one could look for in astrophysical observations. For real situations
one needs to take account the fact that black holes are not simply vacuum solutions. Accretion disks in active galactic 
nuclei constitute one of the most studied questions about ``real'' black hole physics. Another interesting problem 
which concerns rotating charged black holes is their electrical shielding by a charged dust; its solution could
give rise to new astrophysics.\footnote{This problem was pointed out to us by Aldo Treves.} Moreover, other signatures 
could arise from quantum effects like black hole discharge (see for example \cite{damo,belmart,belgcaccia}) or angular 
momentum loss.\\
Among the theoretical models which can provide 
a deeper understanding of the mathematical properties of the field equations for the matter fields
living on the given black hole background and also of thermodynamical contributions of the matter fields to 
black hole thermodynamics the asymptotically AdS case is interesting under many respects, both for the well-known 
relevance of the AdS backgrounds in the AdS-CFT conjecture and in supergravity, and 
for the peculiar thermodynamical properties of AdS black holes, for which the canonical ensemble is well-defined 
\cite{hawpage}.   
In order to avoid the presence of closed time-like curves, 
one has to take into account the universal covering of such an AdS black hole background, which is not globally 
hyperbolic \cite{hawpage}. 
Notwithstanding, physics can still be uniquely defined if essential selfadjointness properties are obtained 
at least for suitable sets of the field parameters like e.g. the mass (cf. \cite{wald}). 
Herein, we limit ourselves to consider the specific problem  of a spin $\frac 12$ 
massive charged Dirac field on the background of an Anti-de Sitter charged rotating black hole. 
Our aim is twofold: on one hand, we pursue the variable separation following 
\cite{Kamran:1984mb} and the analysis of the Hamiltonian description 
in the given background, generalizing the known results for the asymptotically flat Kerr-Newman solution 
\cite{Finster:2000jz};
on the other hand, we consider the problem of essential selfadjointness of the Hamiltonian $\hat H$ 
on $C_0^{\infty} ((r_+,\infty) \times S^2)^4$ for the Dirac field 
on a Kerr-Newman-AdS black hole background in presence also of a magnetic charge. 
This second task is nontrivial, requiring many technical steps in order to be settled. 
The main difficulty arises from the fact that even if for the Hamiltonian version of the 
Dirac equation a Chandrasekhar-like ansatz for variable separation is available, one cannot obtain 
a full reduction of the Hamiltonian into an orthogonal sum of partial wave operators involving only 
the radial variable, and this has to be attributed to the black hole rotation which allows only axial symmetry. 
Nevertheless, 
essential selfadjointness of $\hat H$ on $C_0^{\infty} ((r_+,\infty) \times S^2)^4$
is shown to be equivalent to the essential selfadjointness of another Hamiltonian $\hat H_0$ 
on $C_0^{\infty} ((r_+,\infty) \times S^2)^4$ which is 
defined in a different (but unitarily equivalent) Hilbert space, and a complete reduction into an orthogonal sum 
of partial wave operators involving only the radial variable for $\hat H_0$ is shown to be available. 
This turns out to be useful also for the analysis of qualitative spectral properties of $\hat H$. 
It is worth mentioning that the aforementioned problem associated with variable separation in Chandrasekhar ansatz 
and the occurrence of two Hilbert spaces in the analysis of the Dirac equation 
have been already pointed out  in \cite{Finster:2000jz} for the case of the Dirac equation
on a black hole background of the Kerr-Newman family, which has been considered in 
several studies \cite{finster-rn,finster-axi,yamada,schmid,baticschmid,hafner} (see also \cite{dafermos});  
the above part of our analysis can be of interest also for that case.  
Moreover, our analysis points out some relevant differences to be related to the 
AdS background considered herein, and also introduces some interesting analysis to be associated with 
a magnetically charged black hole: 
we find the Dirac charge quantization as a condition ensuring essential selfadjointness 
for the Hamiltonian operator $\hat H$ on $C_0^{\infty} ((r_+,\infty) \times S^2)^4$.  
Furthermore, as a corollary of our qualitative spectral analysis, we can also conclude that
in the non-extremal case there is no normalizable time-periodic solution of the Dirac equation, in
agreement with the result for the Kerr-Newman case \cite{finster-axi}.\\
The plan of the paper is the following. In section \ref{knads} we describe the background geometry. In
section \ref{diraceq} we consider the equation for the Dirac field.
Section \ref{hamilton} is devoted to the study of the Hamiltonian formulation. In section
\ref{essauto} the essential selfadjointness of the Hamiltonian is analyzed, and in section \ref{eigen-equation}
some spectral properties are deduced. In particular, we can conclude that 
in the non-extremal case the point spectrum is empty. Three appendices complete our analysis.

\section{The Kerr-Newman-AdS solution.}
\label{knads}
We will give a short description of the background geometry underlying our problem. It arises as follows.
One first solves the Einstein-Maxwell equations with a cosmological constant, and next adds
a Dirac field minimally coupled to the electromagnetic field. The Einstein-Maxwell action is
\eqn
S[g_{\mu\nu}, A_\rho]=-\frac 1{16\pi} \int (R-2\Lambda)\sqrt {-\det g} d^4x -\frac 1{16\pi}\int \frac 14 F_{\mu\nu}F^{\mu\nu}
\sqrt {-\det g} d^4x \ ,
\feqn
where $\Lambda=-\frac 3{l^2}$ is the cosmological constant, $R$ the scalar curvature and $F_{\mu\nu}$ the field strength associated to
the potential 1--form $A$:
\eqn
&& F=dA \ , \qquad\ F_{\mu\nu}=\partial_\mu A_\nu -\partial_\nu A_\mu \ ;\\
&& R=g^{\mu\nu} R_{\mu\nu}\ , \qquad\
R_{\mu\nu}=\partial_\rho \Gamma^\rho_{\mu\nu} -\partial_\nu \Gamma^\rho_{\mu\rho}+\Gamma^{\sigma}_{\mu\nu} \Gamma^\rho_{\sigma\rho}
-\Gamma^{\sigma}_{\mu\rho} \Gamma^\rho_{\sigma\nu} \ ,\\
&& \Gamma^\mu_{\nu\rho} =\frac 12 g^{\mu\sigma}(\partial_\nu g_{\sigma\rho}+\partial_\rho g_{\sigma\nu}-\partial_\sigma g_{\nu\rho})\ .
\feqn
The equations of motion are
\eqn
&& R_{\mu\nu}-\frac 12 (R-2\Lambda) g_{\mu\nu}= -2 \left(F_\mu^{\ \ \rho} F_{\nu\rho} -\frac 14 g_{\mu\nu} F_{\rho\sigma}F^{\rho\sigma}\right)
\ , \\
&& \partial_\mu (\sqrt {-\det g} F^{\mu\nu} )=0\ .
\feqn
With respect to a set of vierbein one forms
\eqn
e^i= e^i_\mu dx^\mu\ , \qquad\ i=0,1,2,3\ ,
\feqn
we have
\eqn
ds^2 =g= \eta_{ij} e^i\otimes e^j\ , \qquad\ g_{\mu\nu}=\eta_{ij} e^i_\mu e^j_\nu\ ,
\feqn
where $\eta=\diag (-1,1,1,1)$ is the usual flat Minkowski metric, so that, as usual, we define the $so(1,3)$ valued spin connection one
forms $\omega^i_{\ j}$ such that
\eqn
de^i +\omega^i_{\ j}\wedge e^j=0\ .
\feqn
We will consider the following background solution.\\

The metric is
\eqn
&& ds^2 =-\frac {\Delta_r}{\rho^2}\left[ dt-\frac {a\sin^2 \theta}\Xi d\phi \right]^2 +\frac {\rho^2}{\Delta_r} dr^2
+\frac {\rho^2}{\nat} d\theta^2 +\nat \frac {\sin^2 \theta}{\rho^2} \left[ a dt -\frac {r^2+a^2}\Xi d\phi \right]^2\ ,
\feqn
where
\eqn
&& \rho^2=r^2+a^2\cos^2 \theta\ , \qquad \Xi =1-\frac {a^2}{l^2}\ , \qquad \nar=(r^2+a^2)\left(1+\frac {r^2}{l^2} \right)-2mr+z^2\ ,\\
&& \nat =1-\frac {a^2}{l^2} \cos^2 \theta \ , \qquad z^2=q_e^2 +q_m^2\ ,
\feqn
and the electromagnetic potential and field strength are
\eqn
&& A=-\frac {q_e r}{\rho \snar} e^0 -\frac {q_m \cos \theta}{\rho \snat \sin \theta} e^1\ ,\\
&& F=-\frac {1}{\rho^4}[q_e (r^2-a^2\cos^2\theta)+2q_m ra \cos \theta]e^0 \wedge e^2\cr
&&\qquad\ +\frac {1}{\rho^4}[q_m (r^2-a^2\cos^2\theta)-2q_e ra \cos \theta]e^3 \wedge e^1\ ,
\feqn
where we introduced the vierbein
\eqn
&& e^0 =\frac {\snar}{\rho} \left( dt-\frac {a\sin^2 \theta}\Xi d\phi \right)\ ,\\
&& e^1 =\frac {\snat \sin \theta}\rho \left( a dt -\frac {r^2+a^2}{\Xi} d\phi  \right)\ , \\
&& e^2 =\frac \rho{\snar} dr \ ,\\
&& e^3 =\frac \rho{\snat} d\theta \ .
\feqn
The parameters $m, a, q_e, q_t$ are related to the mass, angular momentum, electric and magnetic charge by the Komar integrals
(see \cite{Caldarelli:1999xj})
\eqn
M=\frac m{\Xi^2}\ , \quad J=\frac {am}{\Xi^2}\ , \quad Q_e=\frac {q_e}\Xi\ , \quad Q_m=\frac {q_m}\Xi\ .
\feqn
We are interested in the case when an event horizon (corresponding to $\nat=0$) appears, that is for $m\geq m_{ext}$,
\eqn
m_{ext}=\frac l{3\sqrt 6}\left(\sqrt {\left(1+\frac {a^2}{l^2}\right)^2+\frac {12}{l^2}(a^2+z^2)}+2\frac {a^2}{l^2}+2 \right)\cr
\qquad\ \times \left(\sqrt {\left(1+\frac {a^2}{l^2}\right)^2+\frac {12}{l^2}(a^2+z^2)}-\frac {a^2}{l^2}-1 \right)^{\frac 12} \ .
\label{mass-extr}
\feqn

\newpage

\section{The Dirac equation.}
\label{diraceq}

The Dirac equation for a charged massive particle of mass $\mu$ and electric charge $e$ is
\eqn
(i\gamma^\mu D_\mu -\mu )\psi=0\ ,
\feqn
where $D$ is the Koszul connection on the bundle $S\otimes U(1)$, $S$ being the spin bundle over the AdS-Kerr-Newman manifold, that is
\eqn
D_\mu =\de_\mu +\frac 14 \omega_\mu^{\ ij} \Gamma_i \Gamma_j +ie A_\mu\ .
\feqn
Here $\omega^{ij}=\omega^i_{\ k} \eta^{kj}$ are the spin connection one forms associated to a vierbein $v^i$, such that $ds^2=\eta_{ij} v^i \otimes v^j$, $\eta$
being the usual Minkowski metric.
$\gamma_\mu$ are the local Dirac matrices, related to the point independent Minkowskian Dirac matrices $\Gamma_i$ by the relations
$\gamma_\mu =v_\mu^i \Gamma_i$.\\
Here we use the representation
\eqn
\Gamma^0 =\left(
\begin{array}{cc}
\mathbb {O} & -\mathbb {I} \\
-\mathbb {I} & \mathbb {O}
\end{array}
\right)\ ,
\qquad
\vec \Gamma =\left(
\begin{array}{cc}
\mathbb {O} & -\vec \sigma \\
\vec \sigma & \mathbb {O}
\end{array}
\right)\ ,
\feqn
where
\eqn
\mathbb {O} =\left(
\begin{array}{cc}
0 & 0 \\
0 & 0
\end{array}
\right)\ ,
\qquad
\mathbb {I} =\left(
\begin{array}{cc}
1 & 0 \\
0 & 1
\end{array}
\right)\ ,
\feqn
and $\vec \sigma$ are the usual Pauli matrices
\eqn
\sigma_1  =\left(
\begin{array}{cc}
0 & 1 \\
1 & 0
\end{array}
\right)\ ,
\qquad
\sigma_2  =\left(
\begin{array}{cc}
0 & -i \\
i & 0
\end{array}
\right)\ ,
\qquad
\sigma_3  =\left(
\begin{array}{cc}
1 & 0 \\
0 & -1
\end{array}
\right)\ .
\feqn
Thus
\eqn
\gamma_\mu \gamma_\nu +\gamma_\nu \gamma_\mu =-2 g_{\mu\nu}\ .
\feqn

We can now separate variables following the general results of \cite{Kamran:1984mb}. Let us introduce the null
Newman-Penrose (symmetric) frame
\eqn
&& \theta^1=\frac 1{\sqrt 2} |Z(r,\theta)|^{\frac 12} \left[ \frac {W(r)}{Z(r,\theta)} \left( dt+\frac {a\sin^2\theta}\Xi d\phi \right)
+\frac {dr}{W(r)} \right]\ ,\\
&& \theta^2=\frac 1{\sqrt 2} |Z(r,\theta)|^{\frac 12} \left[ \frac {W(r)}{Z(r,\theta)} \left( dt+\frac {a\sin^2\theta}\Xi d\phi \right)
-\frac {dr}{W(r)} \right]\ ,\\
&& \theta^3=\frac 1{\sqrt 2} |Z(r,\theta)|^{\frac 12} \left[ \frac {X(\theta)}{Z(r,\theta)} \left( a dt-\frac {r^2+a^2}\Xi d\phi \right)
+i\frac {\sin \theta d\theta}{X(\theta)} \right]\ ,\\
&& \theta^4=\bar \theta^3\ ,
\feqn
where
\eqn
Z(r,\theta)= \frac {r^2+a^2 \cos^2 \theta}{\Xi}\ , \quad W(r)=\frac {\snar}{\Xi^{\frac 12}}\ , \quad
X(\theta)=\frac {\snat \sin \theta}{\Xi^{\frac 12}}\ ,
\feqn
so that
\eqn
ds^2 =-2 (\theta^1 \theta^2 -\theta^3 \theta^4)\ ,
\feqn
and
\eqn
A=-\frac 1{\sqrt {2 |Z(r,\theta)}|} \left[ \frac {H(r)}{W(r)}(\theta^1+\theta^2) + \frac {G(\theta)}{X(\theta)} (\theta^3+\theta^4) \right]\ ,
\feqn
with
\eqn
H(r)=Q_e r \ , \quad G(\theta)=Q_m \cos \theta\ .
\feqn
Note that
\eqn
&& e^0 =\frac 1{\sqrt 2} (\theta^1+\theta^2)\ , \qquad e^1=\frac 1{\sqrt 2} (\theta^3+\theta^4)\ , \\
&& e^3 =\frac 1{\sqrt 2} (\theta^1-\theta^2)\ , \qquad e^4=\frac 1{\sqrt 2 \ i} (\theta^3-\theta^4)\ .
\feqn
The Petrov type D condition ensures the existence of a {\it phase function} ${\mathcal B} (r,\theta)$ such that
\eqn
d{\mathcal B} =\frac 1{4Z(r,\theta)} \left( -2a \frac {\cos \theta}\Xi dr -2ar \frac {\sin \theta}\Xi d\theta \right)\ ,
\feqn
which indeed gives
\eqn
{\mathcal B}(r,\theta)=\frac i4 \log \frac {r-ia\cos\theta}{r+ia\cos\theta} \ .
\feqn
Now let us write the Dirac equation as
\eqn
H_D \psi=0\ .
\feqn
Under a transformation $\psi \mapsto S^{-1}\psi$, with
\eqn
S=Z^{-\frac 14}\diag (e^{i{\mathcal B}},e^{i{\mathcal B}},e^{-i{\mathcal B}},e^{-i{\mathcal B}})\ ,
\feqn
it changes as
\eqn
S^{-1} H_D S (S^{-1}\psi)=0\ .
\feqn
If we multiply this equation times
\eqn
U=iZ^{\frac 12}\diag (e^{2i{\mathcal B}},-e^{2i{\mathcal B}},-e^{-2i{\mathcal B}},e^{-2i{\mathcal B}})\ ,
\feqn
and introduce the new wave function
\eqn
\tilde \psi =(\nat \nar)^{\frac 14} S^{-1} \psi\ , \label{reduced}
\feqn
then the Dirac equation takes the form
\eqn
({\mathcal R}(r)+{\mathcal A}(\theta)) \tilde \psi =0\ , \label{dirac}
\feqn
where
\eqn
&& {\mathcal R}=\left(
\begin{array}{cccc}
i\mu r  & 0 & -\snar {\mathcal D}_+ & 0 \\
0 & -i\mu r & 0 & -\snar {\mathcal D}_- \\
-\snar {\mathcal D}_- & 0 & -i\mu r & 0 \\
0 & -\snar {\mathcal D}_+ & 0 & i \mu r
\end{array}
\right) \ , \\
&& {\mathcal A}=\left(
\begin{array}{cccc}
-a\mu \cos \theta  & 0 & 0 & -i\snat {\mathcal L}_- \\
0 & a\mu \cos \theta  & -i\snat {\mathcal L}_+ & 0  \\
0 & -i\snat {\mathcal L}_- & -a\mu \cos \theta & 0  \\
-i\snat {\mathcal L}_+ & 0 & 0 &  a \mu \cos \theta
\end{array}
\right) \ ,
\feqn
and
\eqn
&& {\mathcal D}_{\pm} =\de_r \pm \frac 1{\nar}\left( (r^2+a^2)\de_t -a\Xi \de_\phi +ie q_e r \right)\ ,\\
&& {\mathcal L}_{\pm} =\de_\theta +\frac 12 \cot \theta \pm \frac i{\nat \sin \theta} \left( \Xi \de_\phi -a\sin^2 \theta \de_t
+ie q_m \cos \theta \right) \ .
\feqn
Separation of variables can then be obtained searching for solutions of the form
\eqn\label{separation}
\tilde \psi (t,\phi, r, \theta)=e^{-i\omega t}e^{-ik \phi}
\left(
\begin{array}{c}
R_1 (r) S_2 (\theta)\\
R_2 (r) S_1 (\theta)\\
R_2 (r) S_2 (\theta)\\
R_1 (r) S_1 (\theta)
\end{array}
\right)\ , \qquad k \in \mathbb{Z}+\frac 12 \ .
\feqn

\section{Hamiltonian formulation.}
\label{hamilton}

The Hamiltonian for the Dirac equation can be read from (\ref{dirac}) rewriting it in the form \cite{Finster:2000jz}
\eqn
i\de_t \tilde \psi =H \tilde \psi\ .
\label{hamilton-dirac}
\feqn
Indeed we find
\eqn
H=\left[ \left(1-\frac {\nar}{\nat} \frac {a^2 \sin^2 \theta}{(r^2+a^2)^2} \right)^{-1}
\left( \mathbb {I}_4 -\frac {\snar}{\snat} \frac {a\sin \theta}{r^2 +a^2} B C \right)\right](\tilde{\mathcal R}+\tilde{\mathcal A})\ ,\label{hamiltonian}
\feqn
where $\mathbb {I}_4$ is the $4\times 4$ identity matrix,
\eqn
&&\tilde{\mathcal R}= -\frac {\mu r \snar}{r^2+a^2}
\left(
\begin{array}{cccc}
0 & 0 & 1 & 0 \\
0 & 0 & 0 & 1 \\
1 & 0 & 0 & 0 \\
0 & 1 & 0 & 0
\end{array}
\right)
+
\left(
\begin{array}{cccc}
{\mathcal E}_- & 0 & 0 & 0 \\
0 & -{\mathcal E}_+ & 0 & 0 \\
0 & 0 & -{\mathcal E}_+ & 0 \\
0 & 0 & 0 & {\mathcal E}_-
\end{array}
\right)\ , \\
&& \tilde{\mathcal A}=\frac {a \mu \cos \theta \snar}{r^2+a^2}
\left(
\begin{array}{cccc}
0 & 0 & i & 0 \\
0 & 0 & 0 & i \\
-i & 0 & 0 & 0 \\
0 & -i & 0 & 0
\end{array}
\right)
+
\left(
\begin{array}{cccc}
0 & -{\mathcal M}_- & 0 & 0 \\
{\mathcal M}_+ & 0 & 0 & 0 \\
0 & 0 & 0 & {\mathcal M}_- \\
0 & 0 & -{\mathcal M}_+ & 0
\end{array}
\right)\ ,\\
&& {\mathcal E}_\pm =i \frac {\nar}{a^2+r^2} \left[ \de_r \mp \frac {a\Xi}{\nar} \de_\phi \pm i \frac {e q_e r}{\nar} \right]\ ,\\
&& {\mathcal M}_\pm =\frac {\snar \snat}{r^2+a^2} \left[ \de_\theta +\frac 12 \cot \theta \pm \frac {i\Xi}{\nat \sin \theta} \de_\phi
\mp \frac {e q_m \cot \theta}{\nat} \right]\ ,
\feqn
and
\eqn
B=
\left(
\begin{array}{cccc}
0 & 0 & -i & 0 \\
0 & 0 & 0 & i \\
i & 0 & 0 & 0 \\
0 & -i & 0 & 0
\end{array}
\right)\ , \qquad
C=\left(
\begin{array}{cccc}
0 & 0 & 0 & i \\
0 & 0 & -i & 0 \\
0 & i & 0 & 0 \\
-i & 0 & 0 & 0
\end{array}
\right)
\feqn
satisfy $[B,C]=0$, $B^2=C^2=\mathbb {I}_4$. Cf. also \cite{Finster:2000jz} for the Kerr-Newman case.   
We need now to specify the Hilbert space. We do it as follows.\\
If we foliate spacetime in $t=constant$ slices ${\mathcal S}_t$, the metric on any slice (considering the shift vectors) is
\eqn
d\gamma^2 =\gamma_{\alpha \beta} dx^\alpha dx^\beta \ ,
\feqn
where $\alpha=1,2,3$ and
\eqn
\gamma_{\alpha\beta} =g_{\alpha\beta}-\frac {g_{0\alpha} g_{0\beta}}{g_{00}} \ ,
\feqn
and local measure
\eqn
d\mu_3= \sqrt {\det \gamma}\ dr d\theta d\phi =\frac {\sin \theta}{\Xi} \frac {\rho^3}{\sqrt {\nar -a^2 \nat \sin^2 \theta}}
\ dr d\theta d\phi \ . \label{measure3}
\feqn
In particular the four dimensional measure factors as
\eqn
d\mu_4= \sqrt {-g_{00}} d\mu_3 dt \ . \label{factor}
\feqn
The action for a massless uncharged Dirac particle is then
\eqn
S=\int_{\mathbb{R}} dt \int_{{\mathcal S}_t} \sqrt {-g_{00}}\ ^t \psi^* \Gamma^0 \gamma^\mu D_\mu \psi d\mu_3 \ ,
\feqn
where the star indicates complex conjugation. Here with ${\mathcal S}_t$ we mean the range of coordinates parameterizing the region external to the
event horizon: $r>r_+$, that is ${\mathcal S}_t:= {\mathcal S} =(r_+, \infty)\times (0,\pi)\times (0,2\pi)$.
Then, the scalar product between wave functions should be
\eqn
\langle \psi | \chi \rangle = \int_{{\mathcal S}} \sqrt {-g_{00}}\ ^t \psi^* \Gamma^0 \gamma^t \chi d\mu_3 \ .
\feqn
We can now use (\ref{measure3}), (\ref{reduced}) and the relation
\eqn
\gamma^2=e^t_0 \Gamma^0 +e^t_1 \Gamma^1\ ,
\feqn
to express the product in the space of reduced wave functions (i.e. (\ref{reduced})):
\eqn
\langle \tilde \psi |\tilde \chi \rangle =\int_{r_+}^{\infty} dr \int_0^\pi d\theta \int_0^{2\pi} d\phi \frac {r^2+a^2}\nar
\frac {\sin \theta}{\snat}\ ^t\tilde \psi^* \left(\mathbb {I}_4 +\frac {\snar}{\snat} \frac {a\sin \theta}{r^2 +a^2} B C \right)\tilde \chi \ ,
\label{product}
\feqn
where we have dropped a factor $\Xi^{-\frac 12}$.
Note that the matrix in the parenthesis is the inverse of the one in the square brackets in (\ref{hamiltonian}), and coincides
with the one introduced in \cite{Finster:2000jz} (improved to the AdS case).\\
We show that the above scalar product is positive definite. With this aim,
being $\pm 1$ the eigenvalues of $BC$, we need to prove that
\eqn
\eta:=\sup_{r\in (r_+,\infty), \theta\in (0,\pi)} \alpha(r,\theta)<1\ ,
\feqn
where
\eqn
\alpha(r,\theta)=\frac {\snar}{\snat} \frac {a\sin\theta}{r^2+a^2}\ .
\feqn
We can write $\alpha(r,\theta)=\beta(r)\gamma(\theta)$, with
\eqn
\gamma(\theta)=\frac {\sin\theta}{\snat}  \ .
\feqn
Then
\eqn
\gamma'(\theta)=\frac {\cos\theta}{\nat^{\frac 32}}\left( 1-\frac {a^2}{l^2} \right)
\feqn
so that $\gamma$ reaches its maximum at $\theta=\pi/2$ and
\eqn
\gamma(\pi/2)=1\ ,
\feqn
which implies
\eqn
\alpha (r,\theta)\leq \beta(r).
\feqn
Next, from
\eqn
0=\nar(r_+)
\feqn
we have
\eqn
z^2-2mr_+=-(r_+^2+a^2)(r_+^2+l^2)/l^2 <0\ ,
\feqn
and then, for $r\geq r_+$ we have $z^2-2mr<0$.
Thus
\eqn
\beta^2 (r)=\frac {a^2\nar}{(r^2+a^2)^2}=\frac {a^2}{l^2}\frac {r^2+l^2}{r^2+a^2}+a^2\frac {z^2-2mr}{(r^2+a^2)^2}\leq
\frac {a^2}{l^2}\frac {r^2+l^2}{r^2+a^2} =h(r)\ .
\feqn
Now, the last function is a decreasing function of $r$, so that for $r\geq r_+>0$ we have $h(r)\leq h(r_+)< h(0)$, so that
\eqn
\beta^2 (r)\leq h(r_+)=\frac {a^2}{l^2}\frac {r_+^2+l^2}{r_+^2+a^2}<h(0)=1\ ,
\feqn
and
\eqn
\eta \leq \sqrt {h(r_+)} <1\ .
\feqn

\section{Essential selfadjointness of \boldmath{$\hat H$}.}
\label{essauto}

Let us introduce the space of functions ${\mathcal L}^2:=(L^2((r_+,\infty) \times S^2; d\mu))^4$ with 
measure 
\eqn
d\mu = \frac {r^2+a^2}\nar \frac {\sin \theta}\snat dr d\theta d\phi.
\feqn
and define ${\cHH}_{<>}$ as the Hilbert space ${\mathcal L}^2$ with the scalar product
(\ref{product}). We will also consider a second Hilbert space
$\cHH_{()}$, which is obtained from  ${\mathcal L}^2$ with the scalar product
\eqn
( \psi | \chi ) =\int_{r_+}^{\infty} dr \int_0^\pi d\theta \int_0^{2\pi} d\phi \frac {r^2+a^2}\nar
\frac {\sin \theta}{\snat}\ ^t \psi^* \chi \ =\int d\mu ^t\psi^* \chi \ .
\label{ri-product}
\feqn
It is straightforward to show that $||\cdot||_{<>}$ and $||\cdot||_{()}$ are equivalent norms. 
It is also useful to introduce $\hat \Omega^2:{\mathcal L}^2\to {\mathcal L}^2$ as the multiplication
operator by $\Omega^2 (r,\theta)$: 
\eqn
\Omega^2 (r,\theta):=\mathbb {I}_4 +\alpha (r,\theta)\;  B C\ . \label{omega2}
\feqn
Then we have
\eqn
\langle \psi | \chi \rangle =\int d\mu ^t\psi^* \Omega^2 \chi \ = ( \psi | \hat \Omega^2 \chi ) \ .
\label{red_prod}
\feqn
We introduce also $\hat \Omega^{-2}:{\mathcal L}^2\to {\mathcal L}^2$ as the multiplication operator by $\Omega^{-2}$:  
\eqn
\Omega^{-2} (r,\theta):=\frac{1}{1-\alpha^2 (r,\theta)}\left(\mathbb {I}_4 -\alpha (r,\theta)\;
B C \right)\ ,  \label{omega-2}
\feqn
and analogously $\hat \Omega,\hat \Omega^{-1}$ are defined as operators from
${\mathcal L}^2$ to ${\mathcal L}^2$ which multiply by $\Omega (r,\theta)$,
$\Omega^{-1} (r,\theta)$ respectively, where $\Omega$ and $\Omega^{-1}$ are defined as
the principal square root of $\Omega^2$ and $\Omega^{-2}$ respectively. 
The following properties are useful for our subsequent analysis. 
As matrix-valued functions,
both $\Omega^2$ and $\Omega^{-2}$ are trivially bounded, and
this holds true also for $\Omega$ and $\Omega^{-1}$. $\hat \Omega^2,\hat \Omega,\hat \Omega^{-2},\hat \Omega^{-1}$
are injective and surjective; moreover, as operators from $\cHH_{()}$ to $\cHH_{()}$, they are 
positive, bounded and selfadjoint. 
Injectivity of $\hat \Omega^2$, $\hat \Omega^{-2}$ can be proven by direct inspection. 
Being $\hat \Omega^2$ injective, also $\hat \Omega$ is injective. 
Surjectivity is also easily deduced. Indeed,
being $\hat \Omega^{-2}$ defined everywhere, $\hat \Omega^2$ is also surjective and then
$\hat \Omega$ is surjective too. The same properties hold true for $\hat \Omega^{-2}$ and
$\hat \Omega^{-1}$.  
Positivity is easily proven by carrying the matrices
$\Omega^2,\Omega,\Omega^{-2},\Omega^{-1}$ into the diagonal form and then
by taking into account that $\sup_{r,\theta} \alpha<1$ for $a^2<l^2$.
Analogously also boundedness is proven. Positivity implies selfadjointness.\\

Let us set $H_0 :=\tilde{\mathcal R}+\tilde{\mathcal A}$, which is formally selfadjoint on $\cHH_{()}$, 
and define the operator $\hat  H_0$ on ${\mathcal L}^2$ 
with 
\eqn
&&D(\hat H_0)= C_0^{\infty} ((r_+,\infty) \times S^2)^4 =:\cDD \cr
&&\hat H_0 \chi = H_0 \chi, \quad \chi \in \cDD.
\feqn
Notice that $\cDD$ is dense in ${\cHH}_{()}$. 
Let us point out that for the formal differential expression $H$ in (\ref{hamiltonian}), 
which is formally selfadjoint on $\cHH_{<>}$, 
one can write $H=\Omega^{-2} H_0$, 
Then we define on ${\mathcal L}^2$ 
the differential operator $\hat H= \hat \Omega^{-2} \hat H_0$, with 
\eqn
&&D(\hat H) = \cDD\cr
&&\hat H \chi = H \chi, \quad \chi \in \cDD.
\label{h-h0}
\feqn
As to the symmetry of $\hat H$ on $\cDD \subset \cHH_{<>}$, we note that for all $f,g\in \cDD$ it holds 
$\langle f | \hat H g \rangle = (f | \hat H_0 g)$ and
$\langle \hat H f | g \rangle = (\hat H f |\hat \Omega^2 g) = (\hat \Omega^2 \hat H f | g) = (\hat H_0 f | g)$,
and then $\hat H$ is symmetric iff
$\hat H_0$ is symmetric on $\cDD \subset \cHH_{()}$. 
Symmetry of $\hat H_0$ is proven by direct inspection:
The only problem
could be the integration by parts in $\theta$ and $r$.
The $r$ derivatives arise in the scalar product from the terms involving the differential
operators ${\mathcal E}_\pm$, and they appear in the following form:
$$
\int \frac {\sin \theta}\snat \ ^t\tilde \psi^* {\mathrm {diag}} (i\de_r,-i\de_r,-i\de_r,i\de_r)
 \tilde \chi dr d\theta d\phi\ ,
$$
where ${\mathrm {diag}} (i\de_r,-i\de_r,-i\de_r,i\de_r)$ stays for the diagonal matrix
whose non vanishing entries are explicitly given.\\
The $\theta$ derivatives come out from ${\mathcal L}_\pm$ and appear in the form
$$
\int {\sin \theta} \frac 1{\snar} \ ^t\tilde \psi^* E \frac 1{\sqrt {\sin \theta}} \de_\theta (\sqrt {\sin\theta} \tilde \chi) dr d\theta d\phi\ ,
$$
where
$$
E=\left(
\begin{array}{cccc}
0 & -1 & 0 & 0 \\
1 & 0 & 0 & 0 \\
0 & 0 & 0 & 1 \\
0 & 0 & -1 & 0
\end{array}
\right)\ .
$$
{F}rom these expressions the symmetry of the operator $\hat H_0$ is easily checked.\\

We prove that there exists an unitary isomorphism between $\cHH_{<>}$ and
$\cHH_{()}$. We follow a line of thought which is strictly analogous
to the one allowing to prove the unitary equivalence between the
Hilbert space $L^2 (a,b,q)$ and $L^2 (a,b)$, where the former space
has measure $q(x) dx$ and $q:(a,b)\to {\mathbb R}$ is a measurable,
almost everywhere positive and locally integrable function. Cf. 
\cite{weidhilb}, pp. 247-248.
\begin{lemma}
The map $V_{\Omega}:\cHH_{<>} \mapsto \cHH_{()}$ defined by
\eqn
(V_{\Omega}
\psi) (r,\theta,\phi)= \Omega (r,\theta) \psi (r,\theta,\phi)
\feqn
is an unitary isomorphism of Hilbert spaces.
\end{lemma}
\begin{proof}
One has
\eqn
(V_{\Omega} \psi | V_{\Omega} \chi) = (\hat \Omega \psi |\hat \Omega \chi) =
(\psi|\hat \Omega^2 \chi) = \langle \psi | \chi \rangle.
\feqn
Then, due to the aforementioned properties of $\hat \Omega$ and $\hat \Omega^{-1}$,
$V_{\Omega}$ is an (unitary) isomorphism of Hilbert spaces \cite{weidhilb}. Note also that
$V_{\Omega}^{\ast}:\cHH_{()} \mapsto \cHH_{<>}$ and that $V_{\Omega}^{\ast}$ acts as a multiplication operator
by $\Omega^{-1} (r,\theta)$.
\end{proof}
It is useful to introduce 
$V_{\Omega} \hat H  V_{\Omega}^{-1}$ which is defined on the domain
$V_{\Omega} \cDD \subset \cHH_{()}$ and which is unitarily
equivalent to the operator $\hat H$ defined on $\cDD\subset
\cHH_{<>}$. Then the problem of the essential selfadjointness of
$\hat H$ in $\cDD\subset \cHH_{<>}$ is equivalent to the problem of
essential selfadjointness of $V_{\Omega} \hat H  V_{\Omega}^{-1}$ on
$V_{\Omega} \cDD \subset \cHH_{()}$. 
By taking into account $\hat H = \hat \Omega^{-2} \hat H_0$ the aforementioned
problem amounts to the essential selfadjointness of
$\hat \Omega^{-1} \hat H_0 \hat \Omega^{-1}$ in $V_{\Omega} \cDD\subset \cHH_{()}$.\\

Now we can prove the following result. 

\begin{theorem}
$\hat H$ is essentially selfadjoint 
if and only if
$\hat H_0$ is essentially selfadjoint. 
\end{theorem}
\begin{proof}
The following results are useful:\\
a) let $\hat A$ be a densely defined operator from ${\mathcal H}_1$ to
${\mathcal H}_2$ and let $\hat B$ be a bounded operator from ${\mathcal H}_2$ to ${\mathcal H}_3$. Then
$(\hat B \hat A)^{\ast}= \hat A^{\ast} \hat B^{\ast}$.\\
b) Let $\hat C$, $\hat D$ be densely defined operators defined from ${\mathcal H}_1$ to
${\mathcal H}_2$ and from ${\mathcal H}_2$ to ${\mathcal H}_3$ respectively. Assume that
$\hat D \hat C$ is densely defined from ${\mathcal H}_1$ to ${\mathcal H}_3$, and assume that
$\hat C$ is injective with $\hat C^{-1}\in {\mathcal B}({\mathcal H}_2,{\mathcal H}_1)$. Then
$(\hat D \hat C)^{\ast}=\hat C^{\ast} \hat D^{\ast}$ (cf. ex. 4.18, p.74 of \cite{weidhilb}).\\
In our case, in order to apply (a) one identifies $\hat B$ with $\hat \Omega^{-1}$ and
$\hat A$ with the product
$\hat H_0 \hat \Omega^{-1}$; moreover, in order to apply (b) one identifies
$\hat C$ with $\hat \Omega^{-1}$ and
$\hat D$ with $\hat H_0$. As a consequence, one finds that the essential selfadjointness condition,
which amounts to
$(\hat \Omega^{-1} \hat H_0 \hat \Omega^{-1})^{\ast}=(\hat \Omega^{-1} \hat H_0
\hat \Omega^{-1})^{\ast \ast}$,
is implemented if and only if $\hat H_0^{\ast}=\hat H_0^{\ast \ast}$, i.e. if and only if $\hat H_0$
is essentially selfadjoint on $\cDD\subset \cHH_{()}$. (Indeed,
from (a) we know that $(\hat \Omega^{-1} \hat H_0 \hat \Omega^{-1})^{\ast}=
(\hat H_0 \hat \Omega^{-1})^{\ast} \hat \Omega^{-1}$, and from (b) one finds
$(\hat H_0 \hat \Omega^{-1})^{\ast} \hat \Omega^{-1}=\hat \Omega^{-1} \hat H_0^{\ast} \hat \Omega^{-1}$.
Analogously, $(\hat \Omega^{-1} \hat H_0 \hat \Omega^{-1})^{\ast \ast}=[(\hat \Omega^{-1} \hat H_0
\hat \Omega^{-1})^{\ast}]^{\ast}=
[\hat \Omega^{-1} \hat H_0^{\ast} \hat \Omega^{-1}]^{\ast}=(\hat H_0^{\ast} \hat \Omega^{-1})^{\ast}
\hat \Omega^{-1}=
\hat \Omega^{-1} \hat H_0^{\ast \ast} \hat \Omega^{-1}$.) 
\end{proof}
To sum up, we have shown that the essential selfadjointness of $\hat H$ on $\cDD\subset \cHH_{<>}$ is
equivalent to the essential selfadjointness of $\hat H_0$ on $\cDD\subset \cHH_{()}$. Furthermore,  
if $\hat T_{H_0}$ is a selfadjoint realization of 
$H_0$, to be defined on $D(\hat T_{H_0}):={\mathfrak D}\subset \cHH_{()}$, we introduce   
$\hat T_{H}:=\hat \Omega^{-2} \hat T_{H_0}$ on $\cHH_{<>}$. One can show analogously that 
$\hat T_{H}$ is a selfadjoint realization of $H$ on ${\mathfrak D}\subset \cHH_{<>}$. 
\newpage \noindent
(Proof: by the above unitary isomorphism (cf. Lemma 1) 
let us define $\tilde{\hat{T}}_{H} = \hat \Omega \hat T_{H} \hat \Omega^{-1} = 
\hat \Omega^{-1} \hat T_{H_0} \hat \Omega^{-1}$ on $\Omega {\mathfrak D}$. One finds easily (cf. Theorem 1) 
$\tilde{\hat{T}}_{H}^{\ast} = \hat \Omega^{-1} \hat{ T}_{H_0}^{\ast} \hat \Omega^{-1}=
\hat \Omega^{-1} \hat T_{H_0} \hat \Omega^{-1}= \tilde{\hat{ T}}_{H}$, and then 
$\tilde{\hat{ T}}_{H}$ is selfadjoint. As a consequence 
$\hat T_{H}= \hat \Omega^{-2} \hat T_{H_0}$ on ${\mathfrak D}\subset \cHH_{<>}$ is selfadjoint too, and  
$\hat \Omega^{-2} \hat T_{H_0} f = H f$, $f\in {\mathfrak D}$).


\subsection{Essential selfadjointness of \boldmath{$\hat H_0$} on  
$\cDD=C_0^{\infty} ((r_+,\infty) \times S^2)^4\subset {\cHH}_{()}$}

It is useful to recall that ${\cHH}_{()}\simeq L^2((r_{+},\infty), \frac{r^2+a^2}{\Delta_r} dr)^2 
\otimes L^2 ((0,\pi), \frac{\sin (\theta)}{\sqrt{\Delta_{\theta}}} d\theta)^2 \otimes L^2 (0,2\pi)$, 
where the scalar product is the ``usual'' one ($(f,g):=\int d\mu f^{\ast} g$; for short, the measure is indicated 
by $d\mu$ and $f,g$ are scalar or vector functions depending on the case). 
We introduce the following unitary operator (cf. \cite{winklmeierthesis}) $V:{\cHH}_{()}\to {\cHH}_{()}$:
\eqn
V=
\frac{1}{\sqrt{2}}
\left(
\begin{array}{cccc}
 0  & -i &  0  &  i \\
 i  &  0 & -i  &  0 \\
 0  & -1 &  0  & -1 \\
-1  &  0 & -1  &  0
\end{array}
\right),
\label{def-V}
\feqn
and then we consider the operator $V \hat H_0  V^{\ast}$ on $V\cDD$. 
The latter operator is particularly suitable for the study of essential selfadjointness 
by means of variable separation. One finds:
\eqn
V \tilde{\mathcal R} V^{\ast}=
\left(
\begin{array}{cccc}
 \frac{1}{2} ({\mathcal E}_--{\mathcal E}_{+})+\mu (r) & 0 & -\frac{1}{2} i ({\mathcal E}_++{\mathcal E}_{-}) &  0\\
0  &  \frac{1}{2} ({\mathcal E}_--{\mathcal E}_{+})+\mu (r) & 0  & -\frac{1}{2} i ({\mathcal E}_++{\mathcal E}_{-}) \\
\frac{1}{2} i ({\mathcal E}_++{\mathcal E}_{-})  & 0 &   \frac{1}{2} ({\mathcal E}_--{\mathcal E}_{+})-\mu (r)  & 0 \\
0  & \frac{1}{2} i ({\mathcal E}_++{\mathcal E}_{-})  & 0  & \frac{1}{2} ({\mathcal E}_--{\mathcal E}_{+})-\mu (r)
\end{array}
\right),
\feqn
where $\mu (r):=\mu \frac{r \sqrt{\Delta_r}}{r^2+a^2}$;
moreover, it holds
\eqn
V \tilde{\mathcal A} V^{\ast}=
\left(
\begin{array}{cccc}
 0  &  0 &  -\mu a \cos(\theta) \frac{\sqrt{\Delta_r}}{r^2+a^2}  &  i {\mathcal M}_+  \\
 0  &  0 &  i {\mathcal M}_{-}  &  \mu a \cos(\theta) \frac{\sqrt{\Delta_r}}{r^2+a^2}  \\
 -\mu a \cos(\theta) \frac{\sqrt{\Delta_r}}{r^2+a^2}  &  i {\mathcal M}_+&  0  &  0 \\
 i {\mathcal M}_{-}  &  \mu a \cos(\theta) \frac{\sqrt{\Delta_r}}{r^2+a^2} &  0  &  0
\end{array}
\right).
\feqn
Then, from the explicit expressions of ${\mathcal E}_{\pm}$ and of ${\mathcal M}_{\pm}$ one
obtains
\eqn
V \tilde{\mathcal R} V^{\ast}=
\left(
\begin{array}{cc}
\frac{1}{r^2+a^2} (i a \Xi \partial_{\phi} +e q_e r +\mu r \sqrt{\Delta_r}) {\mathbb {I}}
& \frac{\Delta_r}{r^2+a^2} \partial_r {\mathbb {I}}\\
-\frac{\Delta_r}{r^2+a^2} \partial_r {\mathbb {I}} &
\frac{1}{r^2+a^2} (i a \Xi \partial_{\phi} +e q_e r -\mu r \sqrt{\Delta_r}) {\mathbb {I}}
\end{array}
\right)
\feqn
and
\eqn
V \tilde{\mathcal A} V^{\ast}=
\frac{\sqrt{\Delta_r}}{r^2+a^2}
\left(
\begin{array}{cc}
{\mathbb {O}} & {\mathbb {U}}\\
{\mathbb {U}} & {\mathbb {O}}
\end{array}
\right),
\feqn
where ${\mathbb {U}}$ is the $2\times 2$ matrix formal differential expression 
\eqn
{\mathbb {U}}=
\left(
\begin{array}{cc}
-\mu a \cos (\theta) & i \sqrt{\Delta_{\theta}} (\partial_{\theta}+\frac{1}{2} \cot ({\theta})
+ g)\\
 i \sqrt{\Delta_{\theta}} (\partial_{\theta}+\frac{1}{2} \cot ({\theta})
- g)) &
\mu a \cos (\theta)
\end{array}
\right),
\feqn
with $g:=i \frac{1}{\Delta_{\theta} \sin (\theta)} \Xi \partial_{\phi}
-\frac{1}{\Delta_{\theta}} q_m e \cot ({\theta})$. We define also $\hat {\mathbb {U}}$ to be a 
differential operator in the Hilbert space 
$L^2 ((0,\pi), \frac{\sin (\theta)}{\sqrt{\Delta_{\theta}}} d\theta)^2 \otimes L^2 (0,2\pi)$, with 
domain $D(\hat {\mathbb {U}})=L\{C_0^{\infty}(0,\pi)^2\times C_0^{\infty}(0,2\pi)\}$ 
($L\{\cdot \}$ stays for the linear hull),
and $\hat {\mathbb {U}} S = {\mathbb {U}} S$ for 
$S\in D(\hat {\mathbb {U}})$.\\ 
As a consequence of the above manipulations, we obtain $V \hat H_0  V^{\ast}$ on $V \cDD$, with 
$V \hat H_0  V^{\ast} \chi =V H_0  V^{\ast} \chi$, $\chi \in \cDD$, where the formal differential 
expression $V H_0  V^{\ast}$ is: 
\eqn
V H_0  V^{\ast}=
\left(
\begin{array}{cc}
\frac{1}{r^2+a^2} (i a \Xi \partial_{\phi} +e q_e r +\mu r \sqrt{\Delta_r}) {\mathbb {I}}
& \frac{\Delta_r}{r^2+a^2} \partial_r {\mathbb {I}} + \frac{\sqrt{\Delta_r}}{r^2+a^2} {\mathbb {U}}\\
-\frac{\Delta_r}{r^2+a^2} \partial_r {\mathbb {I}} + \frac{\sqrt{\Delta_r}}{r^2+a^2} {\mathbb {U}}&
\frac{1}{r^2+a^2} (i a \Xi \partial_{\phi} +e q_e r -\mu r \sqrt{\Delta_r}) {\mathbb {I}}
\end{array}
\right).
\label{red-h0}
\feqn
We consider the subset $\tilde\cDD$ of $\cDD$ which contains finite linear combinations of 
functions of the following form: 
\eqn\label{separa-h}
\chi (r,\theta,\phi)=\varepsilon (\phi) 
\left(
\begin{array}{c}
R_1 (r) S_2 (\theta)\\
R_2 (r) S_1 (\theta)\\
R_2 (r) S_2 (\theta)\\
R_1 (r) S_1 (\theta)
\end{array}
\right)\ , 
\feqn
where $\varepsilon (\phi)\in C_0^{\infty} (0,2\pi)$, $R(r):=\left(\begin{array}{c} R_1 (r)\\
R_2 (r)\end{array}\right)\in C_0^{\infty}(r_+,\infty)^2$ and 
$S (\theta):=\left(\begin{array}{c} S_1 (\theta)\\
S_2 (\theta)\end{array}\right)\in C_0^{\infty}(0,\pi)^2$. Then one obtains
\eqn\label{separa-h0}
V\chi (r,\theta,\phi)=\varepsilon (\phi) 
\frac{1}{\sqrt{2}}
\left(
\begin{array}{c}
-i (R_2 (r)-R_1 (r)) S_1 (\theta)\\
-i (R_2 (r)-R_1 (r)) S_2 (\theta)\\
- (R_2 (r)+R_1 (r)) S_1 (\theta)\\
- (R_2 (r)+R_1 (r)) S_2 (\theta)
\end{array}
\right).
\feqn
The subspace $L_k$ spanned by the eigenfunctions $e^{-i k \phi}$, $k\in 
{\mathbb Z}+\frac{1}{2}$ of the 
the selfadjoint operator $i \partial_{\phi}$ with anti-periodic boundary conditions at $0$ and at $2\pi$ 
is such that 
$L^2((r_{+},\infty), \frac{r^2+a^2}{\Delta_r} dr)^2 
\otimes L^2 ((0,\pi), \frac{\sin (\theta)}{\sqrt{\Delta_{\theta}}} d\theta)^2 \otimes L_k$ is 
a reducing subspace for $V \hat H_0  V^{\ast}$. 
Moreover $i \partial_{\phi}$ and $\hat {\mathbb {U}}$ trivially commute.  
Let $\hat {\mathbb {U}}_k \otimes I_k$ be the operator 
obtained by restricting $\hat {\mathbb U}$ to $C_0^{\infty} (0,\pi)^2 \otimes L_k$ ($I_k$ is the identity 
operator on $L_k$); one finds that 
$\hat {\mathbb {U}}_k$, whose formal differential expression is 
\eqn
{\mathbb {U}}_k=
\left(
\begin{array}{cc}
-\mu a \cos (\theta) & i \sqrt{\Delta_{\theta}} (\partial_{\theta}+\frac{1}{2} \cot ({\theta})+ b_k (\theta))\\
 i \sqrt{\Delta_{\theta}} (\partial_{\theta}+\frac{1}{2} \cot ({\theta})
-  b_k (\theta))) &
\mu a \cos (\theta)
\end{array}
\right),
\feqn
where $b_k (\theta):= \frac{1}{\Delta_{\theta} \sin (\theta)} \Xi k
-\frac{1}{\Delta_{\theta}} q_m e \cot ({\theta})$, 
is essentially selfadjoint on $C_0^{\infty} (0,\pi)^2$ 
for any $k\in {\mathbb Z}+\frac{1}{2}$ for $\frac{q_m e}{\Xi} \in \mathbb{Z}$ (see sect. \ref{sub-U} for details). 
Note also that $\hat {\mathbb U}=\oplus_k {\hat {\mathbb {U}}}_k \otimes I_k$. 
If one considers the selfadjoint extension $\bar{\hat {\mathbb {U}}}_k$ of $\hat {\mathbb {U}}_k$, one can show that 
$\bar{\hat {\mathbb {U}}}_k$ has purely discrete spectrum which is simple (see section 
\ref{spectrumUomega} and see also Appendix \ref{spectrumU}). 

Let us introduce  
the (normalized) eigenfunctions $S_{k;j} (\theta):=\left(
\begin{array}{c}
S_{1\; k;j} (\theta)\\
S_{2\; k;j} (\theta)
\end{array}
\right)
$
of the operator $\bar{\hat {\mathbb {U}}}_k$: 
\eqn
\bar{\hat {\mathbb {U}}}_k
\left(
\begin{array}{c}
S_{1\; k;j} (\theta)\\
S_{2\; k;j} (\theta)
\end{array}
\right)
=
\lambda_{k;j}
\left(
\begin{array}{c}
S_{1 \; k;j} (\theta)\\
S_{2 \; k;j} (\theta)
\end{array}
\right),
\label{red-angular}
\feqn
then ${\mathcal H}_{k,j}:=L^2((r_{+},\infty), \frac{r^2+a^2}{\Delta_r} dr)^2\otimes M_{k,j}$, where $M_{k,j}:=\{F_{k;j}(\theta,\phi)\}$, 
with $F_{k;j}(\theta,\phi):=S_{k;j}(\theta) \frac{e^{-i k \phi}}{\sqrt{2\pi}}$,   
is a reducing subspace for $V \hat H_0  V^{\ast}$. 
Let us define $D_{k,j}:=\tilde \cDD \cap {\mathcal H}_{k,j}$. Then 
$V \hat H_0  V^{\ast}|_{D_{k,j}}$ is such that 
$V \hat H_0  V^{\ast} (V\chi) = \omega (V\chi)$ becomes 
\eqn
\left(
\begin{array}{cc}
\frac{1}{r^2+a^2} (a \Xi k +e q_e r +\mu r \sqrt{\Delta_r}) {\mathbb {I}}
& \frac{\Delta_r}{r^2+a^2} \partial_r {\mathbb {I}} + \frac{\sqrt{\Delta_r}}{r^2+a^2} \lambda_{k;j} {\mathbb {I}} \\
-\frac{\Delta_r}{r^2+a^2} \partial_r {\mathbb {I}} + \frac{\sqrt{\Delta_r}}{r^2+a^2} \lambda_{k;j} {\mathbb {I}}  &
\frac{1}{r^2+a^2} (a \Xi k +e q_e r -\mu r \sqrt{\Delta_r}) {\mathbb {I}}
\end{array}
\right) V\chi = \omega V\chi,
\label{red-radial}
\feqn
which is equivalent to the following $2\times 2$ Dirac system for the radial part:
\eqn
\left(
\begin{array}{cc}
\frac{1}{r^2+a^2} (a \Xi k +e q_e r +\mu r \sqrt{\Delta_r})
& \frac{\Delta_r}{r^2+a^2} \partial_r  + \frac{\sqrt{\Delta_r}}{r^2+a^2} \lambda_{k;j}\\
-\frac{\Delta_r}{r^2+a^2} \partial_r  + \frac{\sqrt{\Delta_r}}{r^2+a^2} \lambda_{k;j} &
\frac{1}{r^2+a^2} (a \Xi k +e q_e r -\mu r \sqrt{\Delta_r})
\end{array}
\right)
\left(
\begin{array}{c}
X_1 (r) \\
X_2 (r)
\end{array}
\right)
=
\omega
\left(
\begin{array}{c}
X_1 (r) \\
X_2 (r)
\end{array}
\right),
\label{reduced-radial}
\feqn
where
\eqn
\left(
\begin{array}{c}
X_1 (r) \\
X_2 (r)
\end{array}
\right)
:=
\frac{1}{\sqrt{2}}
\left(
\begin{array}{c}
-i (R_2 (r)-R_1 (r)) \\
- (R_2 (r)+R_1 (r))
\end{array}
\right).
\feqn
Then we obtain $\cHH_{()}=\oplus_{k,j} {\mathcal H}_{k,j}$, and we also obtain
the following orthogonal decomposition \cite{mlak} (also called partial wave decomposition)  
of the operator
$V \hat H_0  V^{\ast}$:
\eqn
V \hat H_0  V^{\ast}=\bigoplus_{k,j} \hat h_{k,j} \otimes I_{k,j},
\feqn
where $I_{k,j}$ stays for the identity operator on $M_{k,j}$ and $\hat h_{k,j}$, which is defined on  
${\mathcal D}_{k,j}:=C_0^{\infty} (r_+,\infty)^2$, has the following formal expression:
\eqn
h_{k,j}:=
\left(
\begin{array}{cc}
\frac{1}{r^2+a^2} (a \Xi k +e q_e r +\mu r \sqrt{\Delta_r})
& \frac{\Delta_r}{r^2+a^2} \partial_r  + \frac{\sqrt{\Delta_r}}{r^2+a^2} \lambda_{k;j}\\
-\frac{\Delta_r}{r^2+a^2} \partial_r  + \frac{\sqrt{\Delta_r}}{r^2+a^2} \lambda_{k;j} &
\frac{1}{r^2+a^2} (a \Xi k +e q_e r -\mu r \sqrt{\Delta_r})
\end{array}
\right)
\label{reduction}
\feqn
In the following subsections, 
we study essential selfadjointness conditions both for the angular momentum operator
${\mathbb {U}}$ and for the reduced Hamiltonian $\hat h_{k,j}$.
Note that if $\hat h_{k,j}$ is essentially selfadjoint on ${\mathcal D}_{k,j}$, 
then $V \hat H_0  V^{\ast}$ is essentially selfadjoint on the linear hull $L\{{\mathcal D}_{k,j}\otimes 
M_{k,j}; k,j\}$ \cite{weidhilb}.

\subsubsection{Essential selfadjointness of $\hat {\mathbb {U}}_k$}
\label{sub-U}

The operator ${\mathbb {U}}_k$ is formally selfadjoint in
$L^2 ((0,\pi), \frac{\sin (\theta)}{\sqrt{\Delta_{\theta}}} d\theta)^2$ whose elements
are indicated by $S(\theta):=\left(\begin{array}{c} S_1 (\theta)\\
S_2 (\theta)\end{array}\right)$. The conditions for the essential selfadjointness of $\hat {\mathbb {U}}_k$ 
on $C_0^{\infty}(0,\pi)^2$ are determined in the following.\\
By means of the unitary transformation
\eqn
W:=
\left(
\begin{array}{cc}
0 & -i \\
1 & 0
\end{array}
\right)
\label{unitW}
\feqn
the operator $\hat {\mathbb {U}}_k$ has a formal differential expression which is carried into the following form 
which corresponds to a Dirac system \cite{weidmann}:
\eqn
W {\mathbb {U}}_k W^{\ast}=
\left(
\begin{array}{cc}
\mu a \cos (\theta) &  \sqrt{\Delta_{\theta}} (\partial_{\theta}+\frac{1}{2} \cot ({\theta})
- b_k (\theta))\\
\sqrt{\Delta_{\theta}} (-\partial_{\theta}-\frac{1}{2} \cot ({\theta})
- b_k (\theta)) &
-\mu a \cos (\theta)
\end{array}
\right).
\feqn
A further Liouville unitary transformation $R: L^2 ((0,\pi), \frac{\sin (\theta)}{\sqrt{\Delta_{\theta}}} d\theta)^2
\to L^2 ((0,\pi), \frac{1}{\sqrt{\Delta_{\theta}}} d\theta)^2$, 
\eqn
(R S)(\theta): = (\sin (\theta))^{\frac{1}{2}} 
S (\theta)=:\Theta (\theta) 
\label{liouR}
\feqn
(cf. \cite{winklmeierthesis} for the Kerr-Newman case)
allows us to determine for $\lambda \in 
{\mathbb C}$ if the limit point case or the limit circle case is implemented according to Weyl's alternative 
\cite{weidmann} by studying the differential system 
$R W {\mathbb {U}}_k W^{\ast} R^{\ast} \Theta=\lambda \Theta$, i.e.
\eqn
\left(
\begin{array}{cc}
\mu a \cos (\theta) &  \sqrt{\Delta_{\theta}} (\partial_{\theta}
- \frac{\Xi k}{\Delta_{\theta} \sin (\theta)} +
\frac{1}{\Delta_{\theta}} q_m e \cot ({\theta}) )\\
\sqrt{\Delta_{\theta}} (-\partial_{\theta}-
\frac{\Xi k}{\Delta_{\theta} \sin (\theta)}+
\frac{1}{\Delta_{\theta}} q_m e \cot ({\theta}) )&
-\mu a \cos (\theta)
\end{array}
\right)
\Theta=\lambda \Theta.
\label{ang-equa}
\feqn
We shall determine for $\lambda \in 
{\mathbb C}$ if the limit point case or the limit circle case is implemented according to Weyl's alternative 
\cite{weidmann}. 
The above equation amounts to a first order differential system which displays a
first kind singularity both at $\theta = 0$ and at $\theta = \pi$ \cite{hsieh,walter}. In the former case, one can
write
\eqn
\theta \partial_{\theta} \Theta = N \Theta,
\feqn
where the smooth matrix $N$ is regular as $\theta \to 0^+$ and 
\eqn
N_0:=\lim_{\theta \to 0^+} N =
\left(
\begin{array}{cc}
-k +\frac{q_m e}{\Xi} &  0\\
0 & k -\frac{q_m e}{\Xi}
\end{array}
\right).
\label{enne0}
\feqn
Then the eigenvalues of $N_0$ are $\pm \nu$ with $\nu=k -\frac{q_m e}{\Xi}$. One can find two linearly
independent solutions $\Theta_1,\Theta_2$ near $\theta=0$ such that
\eqn
\Theta_1 = \theta^{\nu} h_1 (\theta)
\feqn
and
\eqn
\Theta_2 = \theta^{-\nu} (h_2 (\theta)+\log (\theta)h_3 (\theta))
\feqn
where $h_i (\theta):=\left(\begin{array}{c} h_{1;i} (\theta)\\
h_{2;i} (\theta)\end{array}\right)$ 
are analytic near $\theta=0$ for $i=1,2,3$ and $h_3\neq 0$ only for $2\frac {q_m e}\Xi$
integer \cite{walter}. We recall that $k=n+\frac{1}{2}$.
It is easy to conclude that the limit point case \cite{weidmann} occurs at $\theta=0$ only for
\eqn
n\leq \frac{q_m e}{\Xi} -1\ , \quad\ \mbox{and} \quad\ n\geq \frac{q_m e}{\Xi}\ .
\label{condt0}
\feqn
The study at $\theta=\pi$ is analogous. Let us define $\alpha = \pi -\theta$.
Then it is straightforward to show that also for $\alpha =0$ there is a first kind singularity
by studying
\eqn
\alpha \partial_{\alpha} \Theta = M \Theta,
\feqn
where the smooth matrix $M$ is regular as $\alpha \to 0^+$ and
\eqn
M_0:=\lim_{\alpha \to 0^+} M =
\left(
\begin{array}{cc}
k +\frac{q_m e}{\Xi} &  0\\
0 & -k -\frac{q_m e}{\Xi}
\end{array}
\right).
\label{emme0}
\feqn
Then the eigenvalues of $M_0$ are $\pm \rho_0$, with $\rho_0=k +\frac{q_m e}{\Xi}$. One can find two linearly
independent solutions near $\alpha=0$ as above. Then the limit point case occurs at $\theta=\pi$ only for
\eqn
n\geq -\frac{q_m e}{\Xi} \ , \quad\ \mbox{and} \quad\ n\leq -\frac{q_m e}{\Xi}-1\ .
\label{condtpi}
\feqn
{F}rom (\ref{condt0}) and (\ref{condtpi}) we see that
if $\frac{q_m e}{\Xi} \in \mathbb{Z}$ then the essential selfadjointness in
$C^{\infty}_0 (0,\pi)^2$ is obtained for any $n\in \mathbb{Z}$. See also Appendix \ref{essautoU} for
a further discussion.\\
Note that if $\frac{q_m e}{\Xi} \notin \mathbb{Z}$ then the essential selfadjointness of ${\mathbb U}_k$ is obtained for
\eqn
n\in \mathbb{Z}-\{[-1-|\frac{q_m e}{\Xi}|],[|\frac{q_m e}{\Xi}|]\}\ ,
\feqn
where with $[z]$ we mean the integer part of $z$. Then there would be some $k=n+\frac{1}{2}$, with 
$n\in \{[-1-|\frac{q_m e}{\Xi}|],[|\frac{q_m e}{\Xi}|]\}$, such that 
essential selfadjointness does not occur on $C^{\infty}_0 (0,\pi)^2$. We 
limit ourselves to impose herein for the product
$\frac{q_m e}{\Xi}$ to be integer:  this has a nice interpretation,
because it can be related to the Dirac quantization condition 
(we recall that $\frac{q_m}{\Xi}=Q_m$ is the magnetic charge of the black hole);  
see also Appendix \ref{essautoU}.\\
As a consequence, we have shown that the following result holds:
\begin{theorem}
$\hat{\mathbb {U}}_k$ is essentially self adjoint on $C^{\infty}_0 (0,\pi)^2$ 
for any $k=n+\frac{1}{2}$, $n\in {\mathbb Z}$ iff  $\frac{q_m e}{\Xi} \in \mathbb{Z}$.
\end{theorem}
Note that, for $q_m=0$ one recovers the same condition as for the standard Kerr-Newman case
discussed in Refs. \cite{winklmeierthesis,yamada}.

\subsubsection{Essential selfadjointness of $\hat h_{k,j}$}

The differential expression $h_{k,j}$ is formally selfadjoint in the Hilber space 
$L^2 ((r_{+},\infty), \frac{r^2+a^2}{\Delta_r} dr)^2$.
In order to study the essential selfadjointness of the reduced Hamiltonian in $C^{\infty}_0 (r_{+},\infty)^2$ one
has to check if the limit point case occurs both at the event horizon $r=r_{+}$ and at $r=\infty$. In the
former case, it is useful introducing the following reparameterization of the metric in the non-extremal case:
\eqn
\Delta_r = \frac{1}{l^2} (r-r_{+})(r-r_{-})(r^2+(r_{+}+r_{-})r +r_{+}^2+r_{-}^2+r_{+} r_{-}+a^2+l^2),
\label{reparam-delta}
\feqn
where the parameters $m,z^2,a,l$ are replaced by $r_{+},r_{-},a,l$. 
One easily finds:
\begin{eqnarray*} 
m &=& \frac{1}{2 l^2} (r_{+}+r_{-}) (r_{+}^2+r_{-}^2+a^2+l^2)\\
z^2 &=& \frac{1}{2 l^2} r_{+} r_{-} (r_{+}^2+r_{-}^2+r_{+} r_{-}+a^2+l^2) -a^2.
\end{eqnarray*} 
This is a good reparameterization, indeed the Jacobian $J$ of the transformation is
$$ 
J = \frac{1}{2 l^4} (3 r_{+}^2+r_{-}^2+2 r_{+} r_{-}+a^2+l^2 )(r_{+}^2+3 r_{-}^2+2 r_{+} r_{-}+a^2+l^2)
(r_{+}-r_{-})
$$ 
which is strictly positive for non extremal black holes.\\
It is also evident that in the extremal case, where $r_{-}=r_{+}$ a reparameterization analogous to (\ref{reparam-delta})
is available:
$$
\Delta_r^{{\mathrm{extr}}}=
\frac{1}{l^2} (r-r_{+})^2 (r^2+2 r_{+} r +3 r_{+}^2 +a^2+l^2),
$$
by taking into account that in the extremal case the parameters $m,z^2,a,l$ are no more independent
(cf. e.g. (\ref{mass-extr})).\\
We show that the following result holds:

\newpage

\begin{theorem}
$\hat h_{k,j}$ is essentially selfadjoint on $C^{\infty}_0 (r_{+},\infty)^2$ iff $\mu l \geq \frac{1}{2}$.
\end{theorem}
\begin{proof}
We choose the tortoise coordinate $y$ defined by
\eqn
dy = - \frac{r^2+a^2}{\Delta_r} dr
\label{ytortoise}
\feqn
and choose a free integration constant in such a way that $y\in (0,\infty)$. It holds
$y\to \infty \Leftrightarrow r\to {r_{+}}^+$. Then we get
\eqn
h_{k,j}=
\left(
\begin{array}{cc}
0 & -\partial_y\\
\partial_y & 0
\end{array}
\right)
+ V(r(y)),
\label{hamilton-tortoise}
\feqn
and the corollary to thm. 6.8 p.99 in \cite{weidmann} ensures that the limit point case holds for $h_{k,j}$ 
at $y=\infty$.\\
It is also useful to point out that it holds
\eqn
\lim_{y\to \infty} V(r(y))=
\left(
\begin{array}{cc}
\frac{1}{r_{+}^2+a^2} (a k  \Xi+ e q_e r_{+}) & 0\\
0 & \frac{1}{r_{+}^2+a^2} (a k  \Xi+ e q_e r_{+})
\end{array}
\right):=
\left(
\begin{array}{cc}
\varphi_+ & 0\\
0 & \varphi_+
\end{array}
\right).
\feqn
The only problem can be found at $r=\infty$. The differential equation $h_{k,j} X=\omega X$ amounts
to the following differential system:
\eqn
\partial_r X =
\left(
\begin{array}{cc}
\frac{\lambda_{k;j}}{\sqrt{\Delta_r}} & -\frac {\omega(r^2+a^2)}{\Delta_r} -\frac{\mu r}{\sqrt{\Delta_r}}+\frac{P(r)}{\Delta_r}\\
\frac {\omega(r^2+a^2)}{\Delta_r} -\frac{\mu r}{\sqrt{\Delta_r}}-\frac{P(r)}{\Delta_r} & -\frac{\lambda_{k;j}}{\sqrt{\Delta_r}}
\end{array}
\right)
X
\feqn
where $X(r):=\left(\begin{array}{c} X_1 (r)\\
X_2 (r)\end{array}\right)$ and $P(r)=a k \Xi+ e q_e r$. 
In order to study the behavior of this differential system at $r=\infty$
it is useful to introduce momentarily $x=\frac{1}{r}$. Then one obtains
\eqn
x \partial_x X = G(x) X,
\feqn
where the smooth matrix $G(x)$ is regular as $x \to 0^+$ and
\eqn
\lim_{x\to 0^+} G(x) =
\left(
\begin{array}{cc}
0 & \mu l\\
\mu l & 0
\end{array}
\right).
\feqn
A singularity of the first kind is found, with eigenvalues $w_{\pm}=\pm \mu l$. As in the previous subsection,
we can conclude that the limit point case occurs at $r=\infty$ iff
\eqn
\int_c^{\infty} \frac{dr}{r^2} r^{\pm 2 \mu l} = \infty.
\feqn
For $\mu >0$ as in the physical interesting case, the limit point case occurs for $\mu l \geq \frac{1}{2}$,
which is also the required essential selfadjointness condition for the reduced Hamiltonian.
\end{proof}
If $\mu l <\frac{1}{2}$, there is a 1-parameter family of selfadjoint extensions $\hat t_{{k,j}}$ of $\hat h_{k,j}$ \cite{weidmann}.

\section{The eigenvalue equation.}
\label{eigen-equation}

We limit our considerations 
to the case $\mu l \geq \frac{1}{2}$, for which we know that $\hat H$ is essentially selfadjoint on 
$\cDD =C_0^{\infty} ((r_+,\infty) \times S^2)^4\subset {\cHH}_{<>}$, and then there exists a 
unique selfadjoint extension  $\bar{\hat H}$ on ${\mathfrak D}$. See also the conclusions. 
The following relation holds between the eigenvalue equation for $\bar{\hat H}$ and the differential system one finds
by separating the variables as in the Chandrasekhar-like trick. 
We have that $\bar{\hat H} = \hat \Omega^{-2} \bar{\hat H}_0$.
{F}rom
\eqn
\bar{\hat H} \psi=\omega \psi
\feqn
one obtains (cf. (\ref{def-V})) 
\eqn
V\hat \Omega^{-2} V^{\ast} V \bar{\hat H}_0 V^{\ast} \chi = \omega \chi,
\feqn
where $\chi \in V{\mathfrak D}$. Defining the bounded invertible multiplication operator
$\hat D^{-2}:=V\hat \Omega^{-2} V^{\ast}$ and multiplying on the left by $\hat D^{2}$ both the members of the equation, 
one finds
\eqn
V \bar{\hat H}_0 V^{\ast} \chi = \hat D^{2} \omega \chi.
\feqn
Being $D^{2}={\mathbb {I}}_4 + {\mathcal T}$, where
\eqn
{\mathcal T} =
\left(
\begin{array}{cccc}
 0        &  0      &  0       &  i\alpha (r,\theta) \\
 0        &  0      & -i\alpha (r,\theta) &  0 \\
 0        & i\alpha (r,\theta) &  0       &  0 \\
-i\alpha (r,\theta)  &  0      &  0       &  0
\end{array}
\right)
\feqn
is associated with the multiplication operator $\hat {\mathcal T}$ which is 
bounded  and selfadjoint in ${\mathcal H}_{()}$, 
it follows
\eqn
(V \bar{\hat H}_0 V^{\ast} -\hat{\mathcal T}\omega) \chi = \omega \chi,
\label{pseudoeigenvalue}
\feqn
which is just in the form suitable for variable separation by means of the standard trick. One then obtains
the standard form for the separated equations, and formally (compare with equations (\ref{red-h0}), (\ref{red-angular}),  
(\ref{red-radial}) and (\ref{reduced-radial})) the original eigenvalue problem is transformed into 
the (pseudo-)eigenvalue problem (\ref{pseudoeigenvalue}), in which both the radial part and the angular part
are coupled because the angular momentum operator one obtains by variable separation depends on $\omega$:
\eqn
\bar{\hat {\mathbb {U}}}_{k\; \omega}:=\bar{\hat {\mathbb {U}}}_k+\hat{\mathbb {V}}_{\omega},
\feqn
where $\hat{\mathbb {V}}_{\omega}$ is a bounded selfadjoint operator in $L^2 ((0,\pi), \frac{\sin (\theta)}{\sqrt{\Delta_{\theta}}} d\theta)^2$ and is a multiplication operator by 
\eqn
{\mathbb {V}}_{\omega}=
\left(
\begin{array}{cc}
0 & -\frac{i \omega a \sin (\theta)}{\sqrt{\Delta_{\theta}}}\\
\frac{i \omega a \sin (\theta)}{\sqrt{\Delta_{\theta}}} & 0
\end{array}
\right) =: \omega\; {\mathbb V};
\label{v-omega}
\feqn 
then also the eigenvalues $\lambda_{k;j} (\omega)$ of $\bar{\hat {\mathbb {U}}}_{k\; \omega}$ depend on $\omega$. As a consequence, also the radial
eigenvalue equation depends on $\omega$ through its dependence on $\lambda_{k;j}$. See e.g.
\cite{winklmeierthesis,batic,yamada,schmid,baticschmid} for the Kerr-Newman case. The following system
of coupled eigenvalue equations have to be satisfied simultaneously in $L^2 ((0,\pi), \frac{\sin (\theta)}{\sqrt{\Delta_{\theta}}} d\theta)^2$ 
and in $L^2 ((r_{+},\infty), \frac{r^2+a^2}{\Delta_r} dr)^2$ respectively:
\eqn
\bar{\hat {\mathbb {U}}}_{k\; \omega} S = \lambda S,
\label{angularequation}
\feqn
and
\eqn
\bar{\hat h}_{k,j} X = \omega X.
\label{radialequation}
\feqn
We stress again that,
the Dirac equation (\ref{dirac}) in the Chandrasekhar-like variable separation ansatz (\ref{separation})
reduces to the couple of equations (\ref{angularequation}) and (\ref{radialequation}) and
is equivalent, due to the nature of the operator $\hat D^{2}$, 
to the Hamiltonian eigenvalue equation under the same ansatz.\\ 
In order to focus on the relation between the spectral analysis of the Hamiltonian 
$\bar{\hat H}$ and the spectra of the operators $\hat h_{k,j}$ which are obtained by variable separation of the 
pseudo-eigenvalue equation (\ref{pseudoeigenvalue}), we could also heuristically introduce the following trick. 
Let us consider the 1-parameter family of selfadjoint operators
\eqn
\{\hat H_0^{(z)}:= V \bar{\hat H}_0 V^{\ast} -z \hat{\mathcal T}\}_{z\in {\mathbb R}},
\label{family}
\feqn
to be defined on a dense domain $D(\hat H_0^{(z)})\subset {\mathcal H}_{()}$, which is easily understood 
to be independent from $z$. Indeed, 
$ z \hat{\mathcal T}$ is a bounded perturbation of  $V \bar{\hat H}_0 V^{\ast}$ 
and is infinitesimally small with respect to $V \bar{\hat H}_0 V^{\ast}$ \cite{RSII} and then, on the 
domain $V{\mathfrak D}=D(V \bar{\hat H}_0 V^{\ast})=:D(\hat H_0^{(z)})$,  
$\hat H_0^{(z)}$ is selfadjoint and defines an analytical family of type (A) according to Kato's definition
\cite{kato,RSIV}.\\
Each operator in this family 
admits an orthogonal decomposition 
(note that $\bar{\hat {\mathbb {U}}}_{k\; z}=\bar{\hat {\mathbb {U}}}_k+z \hat{\mathbb {V}}$; 
cf. (\ref{v-omega}); moreover, 
the eigenvalues $\lambda_{k;j} (z)$ of $\bar{\hat {\mathbb {U}}}_{k\; z}$ depend on z): 
\eqn
\hat H_0^{(z)}=\bigoplus_{k,j} \bar{\hat h}_{k,j}^{(z)}\otimes I_{k,j},
\feqn
and we get
\eqn
\sigma (\hat H_0^{(z)})=\overline{\bigcup_{k,j} \sigma (\bar{\hat h}_{k,j}^{(z)})},
\feqn
and in particular 
\eqn
\sigma_{p} (\hat H_0^{(z)})=\bigcup_{k,j} \sigma_{p} (\bar{\hat h}_{k,j}^{(z)})
\feqn
(see e.g. Lemma 7 in \cite{schmidt}).\\ 
In order to get a relation between the spectra of this 1-parameter family and the solutions of 
(\ref{pseudoeigenvalue}), we impose the constraint to include only those 
$\omega$ such that $\omega\in \sigma (\hat H_0^{(\omega)})$, which implement (\ref{pseudoeigenvalue}).\\ 
Note also that, for a non rotating black hole, $a=0$ implies that ${\mathcal T}=0$ and 
$\Omega^2=\mathbb {I}_4$. Of course, there is no need to introduce the above 1-parameter family of 
operators for the study of the spectrum, and $\bar{\hat H}=\bar{\hat H}_0$.

In the following we show that the spectrum of the angular momentum operator
$\bar{\hat {\mathbb {U}}}_{k\; \omega}$ is discrete for any $\omega\in \mathbb{R}$. 
Moreover, we show that in the non-extremal case
the radial Hamiltonian $\bar{\hat h}_{k,j}$ for any $\lambda_{k;j}$ has a 
spectrum is absolutely continuous and coincides with $\mathbb{R}$, and then in the latter case we infer that  
no eigenvalue of $\bar{\hat{H}}$ exists. 

\subsection{Spectrum of the operator $\bar{\hat {\mathbb {U}}}_{k\; \omega}$}
\label{spectrumUomega}

We consider the equation ${\mathbb {U}}_{k\; \omega} S-\lambda S=0$. As in \cite{weidmann,weidz}, 
we look for real solutions for real $\lambda$ and 
(cf. \cite{weidmann}, p. 242) we introduce an analogue to Pr\"ufer transformation in the case of Dirac system.
We implement the unitary transformations (\ref{unitW}) and (\ref{liouR}) 
and obtain $R W\bar{\hat{\mathbb U}}_{k\; \omega} W^{\ast} R^{\ast}$. 
Let us define the unitary matrix (cf. \cite{winklmeierthesis} for the Kerr-Newman case without magnetic charge)
\eqn
U:=\frac{1}{\sqrt{2}}
\left(
\begin{array}{cc}
 1 & 1\cr
-1 & 1
\end{array}
\right)
\feqn
and also let us define, thanks to the formal differential expression 
$R W {\mathbb U}_{k\; \omega} W^{\ast} R^{\ast}$ (see eq. (\ref{formalurwu}) below), the following couple of selfadjoint 
operators: 
\eqn
&&D({\mathcal U}_0)=\{ \Theta\in L^2 ((0,c), \frac{d\theta}{\sqrt{\Delta_{\theta}}})^2;
\Theta\; \hbox{is locally absolutely continuous}; B (\Theta) =0;
{\mathcal U}_0 \Theta\in L^2 ((0,c), \frac{d\theta}{\sqrt{\Delta_{\theta}}})^2 \}\cr
&&{\mathcal U}_0 \Theta = R W {\mathbb U}_{k\; \omega} W^{\ast} R^{\ast}\; \Theta, 
\quad \Theta \in D({\mathcal U}_0);\cr
&&D({\mathcal U}_{\pi})=\{ \Theta\in L^2 ((c,\pi), \frac{d\theta}{\sqrt{\Delta_{\theta}}})^2;
\Theta\; \hbox{ is locally absolutely continuous}; B (\Theta) =0;
{\mathcal U}_{\pi} \Theta\in L^2 ((c,\pi), \frac{d\theta}{\sqrt{\Delta_{\theta}}})^2 \},\cr
&&{\mathcal U}_{\pi} \Theta = R W {\mathbb U}_{k\; \omega} W^{\ast} R^{\ast}\; \Theta, 
\quad \Theta \in D({\mathcal U}_{\pi}).
\feqn
$0<c<\pi$ is an arbitrary (regular) point at which the boundary condition
$B (\Theta) = \sin (\beta) \Theta_1 (c) + \cos (\beta) \Theta_2 (c)=0$, with $\beta\in [0,\pi)$ and 
with $\Theta(\theta):=\left(\begin{array}{c} \Theta_1 (\theta)\\
\Theta_2 (\theta)\end{array}\right)$, 
is imposed.\\
One has
\eqn
U  R W {\mathbb U}_{k\; \omega} W^{\ast} R^{\ast}U^{\ast} =
\left(
\begin{array}{cc}
 0 & \sqrt{\Delta_{\theta}} \partial_{\theta}\cr
- \sqrt{\Delta_{\theta}} \partial_{\theta} & 0
\end{array}
\right)+
M(\theta),
\feqn
with
\eqn
M(\theta)=
\left(
\begin{array}{cc}
\frac{\Xi \sigma (\theta)}{\sqrt{\Delta_{\theta}}} \frac{1}{\sin (\theta)}+
\frac{a \omega \sin (\theta)}{\sqrt{\Delta_{\theta}}}
& -\mu a \cos (\theta)\cr
-\mu a \cos (\theta) & -\frac{\Xi \sigma (\theta)}{\sqrt{\Delta_{\theta}}} \frac{1}{\sin (\theta)}+
\frac{a \omega \sin (\theta)}{\sqrt{\Delta_{\theta}}}
\end{array}
\right),
\label{formalurwu}
\feqn
where (being $d :=\frac{q_m e}{\Xi}\in {\mathbb Z}$) 
\eqn
\sigma (\theta) := d  \cos (\theta) -(n+\frac{1}{2}), \quad d,n\in {\mathbb Z}.
\feqn
We can rewrite $U  R W {\mathbb U}_{k\; \omega} W^{\ast} R^{\ast} U^{\ast}$ in the following form:
\eqn
{\mathfrak R}^{-1}(\theta) \left[\left(\begin{array}{cc}
 0 & \partial_{\theta}\cr
- \partial_{\theta} & 0
\end{array}
\right)+{\mathfrak R}(\theta)
\left(
\begin{array}{cc}
\frac{\Xi \sigma (\theta)}{\sqrt{\Delta_{\theta}}} \frac{1}{\sin (\theta)}+
\frac{a \omega \sin (\theta)}{\sqrt{\Delta_{\theta}}}
& -\mu a \cos (\theta)\cr
-\mu a \cos (\theta) & -\frac{\Xi \sigma (\theta)}{\sqrt{\Delta_{\theta}}} \frac{1}{\sin (\theta)}+
\frac{a \omega \sin (\theta)}{\sqrt{\Delta_{\theta}}}
\end{array}
\right) \right],
\feqn
where
\eqn
{\mathfrak R}^{-1}(\theta) =\left(\begin{array}{cc}
 \sqrt{\Delta_{\theta}} & 0\cr
0 & \sqrt{\Delta_{\theta}}
\end{array}
\right).
\feqn
As in \cite{weidmann,weidz}, we can define
\eqn
G(\theta,\lambda)=\lambda {\mathfrak R}(\theta)-M(\theta),
\feqn
and
\eqn
\Theta(\theta) = \bar{\rho} (\theta)
\left(
\begin{array}{c}
\cos \eta (\theta) \\
\sin \eta (\theta)
\end{array}
\right)
\feqn
where
\eqn
\bar{\rho} (\theta) = \sqrt{\Theta^2_1 (\theta)+\Theta^2_2 (\theta)}
\feqn
and
\eqn
\eta (\theta)=\Bigg\{
\begin{array}{c}
\arctan \frac{\Theta_2 (\theta)}{\Theta_1 (\theta)}\quad \hbox{for}\; \Theta_1 (\theta)\not = 0 \\
{\mathrm{arccot}}\; \frac{\Theta_1 (\theta)}{\Theta_2 (\theta)}\quad \hbox{for}\; \Theta_2 (\theta)\not = 0
\end{array}
\feqn
are defined for real solutions of the eigenvalue equation and are absolutely continuous \cite{weidz}.
Then following \cite{weidmann,weidz} one obtains the following differential equation for $\eta (\theta)$:
\eqn
\frac{d}{d\theta} \eta (\theta,\lambda)=H(\theta,\eta(\theta),\lambda),
\label{derieta}
\feqn
where
\eqn
H(\theta,\eta(\theta),\lambda):=
\left(G(\theta,\lambda) \left(
\begin{array}{c}
\cos \eta (\theta) \\
\sin \eta (\theta)
\end{array}
\right) \bigg| \left(
\begin{array}{c}
\cos \eta (\theta) \\
\sin \eta (\theta)
\end{array}
\right)\right),
\feqn
and with $(\cdot|\cdot)$ here we mean the usual Euclidean product in $\mathbb{C}^2$.
One then finds
\eqn
H(\theta,\eta(\theta),\lambda)&=&
\frac{\lambda}{\sqrt{\Delta_{\theta}}}+(2 a \mu \cos (\theta)) \sin (\eta (\theta)) \cos (\eta (\theta))\cr
&+&\left[ \frac{\Xi \sigma (\theta)}{\sqrt{\Delta_{\theta}}} \frac{1}{\sin (\theta)} + \frac{a \omega \sin (\theta)}{\sqrt{\Delta_{\theta}}}
\right](\sin^2 (\eta (\theta))-\cos^2 (\eta (\theta))). 
\feqn
Note that the function $H(\theta,t,\lambda)$ is smooth for any $(\theta,t,\lambda)\in (0,\pi)\times
{\mathbb R}\times {\mathbb R}$. 
We consider first the case of the operator ${\mathcal U}_{\pi}$.
Let us define as in \cite{weidmann}
\eqn
n_+ (\lambda_1,\lambda_2) &=& \liminf_{\theta \to \pi^-} \frac{1}{\pi} \left(
\eta (\theta,\lambda_2) - \eta (\theta,\lambda_1)\right)\cr
n_- (\lambda_1,\lambda_2) &=& \limsup_{\theta \to \pi^-} \frac{1}{\pi} \left(
\eta (\theta,\lambda_2) - \eta (\theta,\lambda_1)\right)\cr
M (\lambda_1,\lambda_2) &=& \dim \left(
E (\lambda_2) - E (\lambda_1) \right),
\feqn
where $E (\lambda)$ is the spectral resolution of 
${\mathcal U}_{\pi}$. It holds $n_- (\lambda_1,\lambda_2)-2 \leq
M (\lambda_1,\lambda_2) \leq n_+ (\lambda_1,\lambda_2) + 2$. In particular, $\lambda$
belongs to the essential spectrum iff for every $\epsilon>0$ it holds
\eqn
n_+ (\lambda - \epsilon,\lambda + \epsilon)=\infty
\feqn
(cf. \cite{weidmann}, p. 248). We now prove the following result.

\begin{theorem}
The essential spectrum of $\bar{\hat {\mathbb U}}_{k\; \omega}$ is empty.
\label{essUomega}
\end{theorem}
\begin{proof}
We start by considering ${\mathcal U}_{\pi}$.
Our aim is to show that the function $\eta (\theta,\lambda)$ is a bounded
function for any finite $\lambda$. As a consequence, the essential spectrum of ${\mathcal U}_{\pi}$
is empty. With this aim, we assume that $\eta$ is unbounded. We show that this assumption leads to
a contradiction both if $\eta$ has no upper bound as $\theta\to \pi^-$ and if
$\eta$ has no lower bound as $\theta\to \pi^-$.\\
We first note that the function
\eqn
\zeta (\theta):=\frac{\sigma (\theta)}{\sqrt{\Delta_{\theta}}} \frac{1}{\sin (\theta)}
\feqn
in a suitable left neighborhood of $\pi$ is either increasing or decreasing according to the sign of
$\sigma (\theta)$ at $\pi$ (its derivative leading term is
$-\cos (\theta) \frac{\sigma (\theta)}{\sqrt{\Delta_{\theta}}} \frac{1}{\sin^2 (\theta)}$
in such a neighborhood).
Let us assume that $\sigma (\pi)>0$. Then, for $\theta_0 < \theta<\pi$ the function $\sigma (\theta)$ is
positive and for $\theta_1<\theta<\pi$, where $\theta_0\leq \theta_1$, the function $\zeta$ is increasing
without upper bound. Then we can choose a $\theta_2>\theta_1$ such that for any $\theta_2<\theta<\pi$ it holds
\eqn
\frac{\Xi \sigma (\theta)}{\sqrt{\Delta_{\theta}}} \frac{1}{\sin (\theta)} + \frac{a \omega \sin (\theta)}{\sqrt{\Delta_{\theta}}}> |\frac{\lambda}{\sqrt{\Xi}}|+2 |a \mu|.
\label{diverg}
\feqn
If $\eta$ has no upper bound,
then there exists a $\theta_3>\theta_2$ such that
\eqn
\eta (\theta_3,\lambda) = p \pi \quad \hbox{for}\ p\in {\mathbb Z}
\feqn
and
\eqn
\eta (\theta,\lambda)>\eta (\theta_3,\lambda)
\feqn
for $\theta_3<\theta<\theta_4<\pi$, i.e. in a suitable right neighborhood of $\theta_3$ the function 
$\eta (\theta,\lambda)$ has to be increasing; if $\eta$ has no lower bound,
then there exists a $\theta_5>\theta_2$ such that
\eqn
\eta (\theta_5) = (p+\frac{1}{2}) \pi \quad \hbox{for}\ p\in {\mathbb Z}
\feqn
and
\eqn
\eta (\theta,\lambda)<\eta (\theta_5,\lambda)
\feqn
for $\theta_5<\theta<\theta_6<\pi$.
In both cases a contradiction is achieved, indeed the function $H(\theta,\eta (\theta),\lambda)$
is negative in a suitable neighborhood of $\theta_3$ in the former case because of (\ref{diverg}):
\eqn
H(\theta_3,\eta (\theta_3),\lambda)=\frac{\lambda}{\sqrt{\Delta_{\theta_3}}}-
\left[\Xi \frac{\sigma(\theta_3)}{\sqrt{\Delta_{\theta_3}}}\frac{1}{\sin (\theta_3)}+
\frac{a \omega \sin (\theta_3)}{\sqrt{\Delta_{\theta_3}}} \right]<0.
\feqn
Then $\eta$ cannot be unbounded from above (cf. (\ref{derieta})).\\
In the latter case one finds
$H(\theta_5,\eta (\theta_5),\lambda)>0$ and then $\eta$ cannot be unbounded from below.
In the case $\sigma (\pi)<0$, one gets the same contradiction
in an analogous way.\\
Analogously, one can conclude that $\sigma_e ({\mathcal U}_0)=\emptyset$. The decomposition method \cite{weidmann} ensures that
$\sigma_e (U R W\bar{\hat {\mathbb U}}_{k\; \omega} W^{\ast} R^{\ast} U^{\ast})=\sigma_e ({\mathcal U}_0)\cup \sigma_e ({\mathcal U}_{\pi})$,
and then we can conclude that the spectrum of $\bar{\mathbb U}_{k\; \omega}$ is discrete. 
On the grounds of theorem 10.8 in
\cite{weidmann}, we can also conclude that the spectrum is simple.
\end{proof}
See also Appendix \ref{spectrumU} for an alternative proof. Note that for $\omega=0$ Theorem 
\ref{essUomega} implies 
that the spectrum of $\bar{\mathbb U}_k$ is discrete.

\subsection{Spectrum of the operator $\bar{\hat h}_{k,j}$}

In order to study the spectral properties of $\bar{\hat h}_{k,j}$, we introduce, as in the previous subsection,
two auxiliary selfadjoint operators $\hat h_{hor}$ and $\hat h_{\infty}$: 
\eqn
&&D(\hat h_{hor})=\{ X\in L^2_{(r_{+},r_0)},\; X
\hbox{ is locally absolutely continuous}; B(X)=0;\; 
\hat h_{hor} X \in L^2_{(r_{+},r_0)}\},\cr
&& \hat h_{hor} X = h_{k,j} X;\cr
&&D(\hat h_{\infty})\hphantom{o}=\{ X\in L^2_{(r_0,\infty)},\; X
\hbox{ is locally absolutely continuous};  B(X)=0;\; 
\hat h_{\infty} X \in L^2_{(r_0,\infty)}\}\cr
&& \hat h_{\infty} X = h_{k,j} X.
\feqn
$r_0$ is an arbitrary point with $r_{+}<r_0<\infty$, at which the boundary condition $B(X):=
\sin (\beta) X_1 (r_0)+\cos (\beta) X_2 (r_0) =0$, with $X(r):=\left(\begin{array}{c} X_1 (r)\\
X_2 (r)\end{array}\right)$ and with $\beta\in [0,\pi)$ is imposed. We also have defined $L^2_{(r_{+},r_0)}:= 
L^2 ((r_{+},r_0), \frac{r^2+a^2}{\Delta_r} dr)^2$ and 
$L^2_{(r_0,\infty)}:=L^2 ((r_0,\infty), \frac{r^2+a^2}{\Delta_r} dr)^2$. 
Note that we omit the indices $k,j$ for these operators.\\ 
We first show that $\hat h_{\infty}$ has discrete spectrum and that in the
non-extremal case $\hat h_{hor}$ has absolutely continuous spectrum, and then we deduce 
qualitative spectral properties for $\bar{\hat h}_{k,j}$. 

On the grounds of  the analysis in \cite{hs-ac,weidz-ac}, we can conclude that the spectrum of $\bar{\hat h}_{k,j}$
is absolutely continuous in the non-extremal case. As to the extremal manifold, we limit
ourselves to point out that $\sigma_e (\bar{\hat h}_{\infty})=\emptyset$ holds true, too; a 
weaker conclusion can instead be stated about the spectrum of $\hat h_{hor}$: it is
absolutely continuous in ${\mathbb{R}}-{\varphi_+}$.\\

We introduce the tortoise coordinate
\eqn
\frac{dx}{dr}=\frac{r^2+a^2}{\Delta_r}
\label{tortoisex}
\feqn
and choose the integration constant in such a way that $r\in (r_+,\infty)$ iff $x\in (-\infty,0)$.
We also point out that, for $r\to \infty$,
i.e. for $x\to 0^-$ one finds $r\sim -\frac{l^2}{x}$. We get
\eqn
\hat h_{k,j}=
\left(
\begin{array}{cc}
0 & \partial_x\\
-\partial_x & 0
\end{array}
\right)
+ V(r(x)).
\feqn
We consider $\hat {h}_{\infty}=\hat h_{k,j}|_{[x(r_0),0)}$.
We introduce the Pr\"ufer-like transformation as in the case of the angular momentum
operator
\eqn
\bar{\rho} (x) = \sqrt{X^2_1 (x)+X^2_2 (x)}
\feqn
and
\eqn
\eta (x)=\Bigg\{
\begin{array}{c}
\arctan \frac{X_2 (x)}{X_1 (x)}\quad \hbox{for}\; X_1 (x)\not = 0\\
{\mathrm{arccot}}\; \frac{X_1 (x)}{X_2 (x)}\quad \hbox{for}\; X_2 (x)\not = 0.
\end{array}
\feqn
We can also define
\eqn
G(x,\omega):=\omega {\mathbb I}-V(r(x)),
\feqn
and obtain the differential equation
\eqn
\frac{d}{dx} \eta (x,\omega)=H(x,\eta(x),\omega),
\label{derietx}
\feqn
where
\eqn
H(x,\eta(x),\omega):=
\left( G(x,\omega) \left(
\begin{array}{c}
\cos \eta (x) \\
\sin \eta (x)
\end{array}
\right) \bigg| \left(
\begin{array}{c}
\cos \eta (x) \\
\sin \eta (x)
\end{array}
\right)\right).
\feqn
One then finds (the dependence of $r$ on $x$ is left implicit)
\eqn
H(x,\eta(x),\omega)&=&
\omega -\frac{a\Xi k+e q_e r}{r^2+a^2}
-\frac{2 \lambda_{k;j} \sqrt{\Delta_r}}{r^2+a^2} \sin (\eta (x,\omega)) \cos (\eta (x,\omega))\cr
&+&\frac{\mu r \sqrt{\Delta_r}}{r^2+a^2} \left( \sin^2 (\eta (x,\omega))- \cos^2 (\eta (x,\omega) \right).
\feqn
Note that the function $H(x,s,\omega)$ is smooth for $(x,s,\omega)\in (-\infty,0)\times
{\mathbb R}\times {\mathbb R}$.

\begin{lemma}
$\sigma_e (\hat h_{\infty})=\emptyset$.
\end{lemma}
\begin{proof}
The leading term in the potential is
proportional to the mass $\mu$ and is monotonically increasing in a suitable left neighborhood $x_0<x<0$ of
$x=0$. We can also find an $x_1\in [x_0,0)$ such that for $x_1<x<0$ one gets
\eqn
\mu \frac{r\sqrt{\Delta_r}}{r^2+a^2}>|\omega|+\frac{|a\Xi k+e q_e r|}{r^2+a^2}+2 |\lambda_{k;j}|
\frac{\sqrt{\Delta_r}}{r^2+a^2}.
\label{bound-x}
\feqn
If $\eta$ is not bounded from above, 
we can find an $x_2\in (x_1,0)$ such that $\eta (x_2,\omega)=p \pi$,
with $p\in {\mathbb Z}$, and $\eta (x,\omega)>\eta (x_2,\omega)$ for $x_2<x<x_3<0$. If $\eta$ is not bounded from
below, we can find an $x_4\in (x_1,0)$ such that $\eta (x_4,\omega)=(q+\frac{1}{2}) \pi$,
with $q\in {\mathbb Z}$, and $\eta (x,\omega)<\eta (x_4,\omega)$ for $x_4<x<x_5<0$. But due to (\ref{bound-x}) 
$H(x,\eta(x),\omega)$ would be negative in a neighbourhood of $x_2$ and it would be positive
in a neighbourhood of $x_4$, against the assumption of an unbounded $\eta$. As a consequence, $\eta$ has to
be bounded, and then  the essential spectrum of $\hat h_{\infty}$ is empty.
\end{proof}
See also Appendix \ref{spectrumhinf} for an alternative proof. Note that, in the case $\mu l < \frac{1}{2}$, 
any selfadjoint extension of $\hat h_{\infty}$ obtained by imposing separated boundary conditions 
at $r_0$ and at $r=\infty$ still has discrete spectrum \cite{weidmann}. It is also remarkable that 
$\sigma_e (\hat h_{\infty})=\emptyset$ is not verified in the standard Kerr-Newman case. 

As to the spectral properties of $\hat h_{hor}$, a suitable change of coordinates consists in
introducing a tortoise-like coordinate defined by eqn. (\ref{ytortoise}).
It is then easy to show that the following result holds.
\begin{lemma}
$\sigma_{ac} (\hat h_{hor}) = {\mathbb R}$ in the non-extremal case.
\end{lemma}
\begin{proof}
The hypotheses of theorem 1 p. 185 of \cite{hs-ac} are verified.
Equivalently in our case we can appeal to theorem 16.7 of \cite{weidmann}, and we find that
the spectrum of $\hat h_{hor}$ is absolutely continuous in ${\mathbb R}-\{\varphi_+\}$. 
This can be proved as follows. Let us write the potential $V(r(y))$ in (\ref{hamilton-tortoise})
\eqn
V(r(y))=\left(
\begin{array}{cc}
\varphi_+  & 0\cr
0 & \varphi_+
\end{array}
\right)+P_2 (r(y)),
\label{p1p2}
\feqn
which implicitly defines $P_2 (r(y))$. The first term on the left of \ref{p1p2} is of course of
bounded variation; on the other hand, $|P_2 (r(y))|\in L^1 (c,\infty)$, with $c\in (0,\infty)$. As
a consequence, the hypotheses of theorem 16.7 in \cite{weidmann} are trivially satisfied, and one
finds that the spectrum of $\hat h_{hor}$ is absolutely continuous in  ${\mathbb R}-\{\varphi_+\}$.\\
We have only to exclude that $\varphi_+$ is not an eigenvalue of $\hat h_{hor}$ (and of the
radial Hamiltonian). As in the Kerr-Newman case (cf. e.g. \cite{yamada}), one needs simply to
replace $\omega$ with $\varphi_+$ and study the asymptotic behavior of the solutions of the linear system
\eqn
X' =
\left(
\begin{array}{cc}
-\lambda_{k;j} \frac{\sqrt{\Delta_r}}{r^2+a^2}
& \varphi_+ - \frac{1}{r^2+a^2} (a \Xi k +e q_e r -\mu r \sqrt{\Delta_r})\\
\frac{1}{r^2+a^2} (a \Xi k +e q_e r +\mu r \sqrt{\Delta_r})-\varphi_+  &
\lambda_{k;j} \frac{\sqrt{\Delta_r}}{r^2+a^2}
\end{array}
\right) X =:\bar{R}(r(y)) X,
\label{eqradialphi}
\feqn
where $r=r(y)$ and where the prime indicates the derivative with respect to $y$.
One easily realizes that in the non-extremal case
\eqn
\int_c^{\infty} dy |\bar{R}(r(y))|<\infty,
\feqn
and then according to the Levinson theorem (see e.g. \cite{eastham}, Theorem 1.3.1 p.8)
one can find two linearly independent
asymptotic solutions as $y\to \infty$ whose leading order is given by $X_I =\left(
\begin{array}{c} 1\\ 0 \end{array} \right)$ and $X_{II} =\left(
\begin{array}{c} 0\\ 1 \end{array} \right)$. As a consequence no normalizable solution of the
equation \ref{eqradialphi} can exists, and then $\varphi_+$ cannot be an eigenvalue. 
\end{proof}
Note that in the non-extremal case, theorem 16.7 in \cite{weidmann} applies also to $\bar{\hat h}_{k,j}$. 
The following result holds:
\begin{theorem}
$\sigma_{ac} (\bar{\hat h}_{k,j}) = {\mathbb R}$ in the non-extremal case. Moreover, the spectrum is simple.
\end{theorem}
\begin{proof}
As a consequence of 
Lemma 2 and Lemma 3 (cf. also remark (2),
pp. 211-212 of \cite{hs-ac}), we can conclude that the spectrum of $\hat h_{k,j}$
is absolutely continuous and coincides with ${\mathbb R}$. On the grounds of theorem 10.8 in
\cite{weidmann}, we can also conclude that the spectrum is simple.
\end{proof}

As to the extremal case, we limit ourselves to state the following result.
\begin{lemma}
The spectrum of $\bar{\hat h}_{k,j}$ is absolutely continuous in ${\mathbb R}-\{\varphi_+\}$ in
the extremal case.
\end{lemma}
\begin{proof}
We refer to theorem 1 p. 185 in \cite{hs-ac}. By using the tortoise coordinate (\ref{ytortoise}) we
rewrite the radial eigenvalue equation (\ref{radialequation}) in the form
\eqn
X' =
\left(
\begin{array}{cc}
p(r(y)) & \omega  + p_2 (r(y))\\
-\omega - p_1 (r(y)) & -p(r(y))
\end{array}
\right) X,
\label{eqradial-extr}
\feqn
where the prime stays for the derivative with respect to $y$,
$p(r):=-\frac{\lambda_{k;j} \sqrt{\Delta_r}}{r^2+a^2}$,
$p_1 (r):=-\frac{1}{r^2+a^2} ( a \Xi k + e q_e r +\mu r \sqrt{\Delta_r} )$ and
$p_2 (r):=-\frac{1}{r^2+a^2} ( a \Xi k + e q_e r -\mu r \sqrt{\Delta_r} )$.\\
Following the notation in \cite{hs-ac}, we define $p(r(y)):=\Delta_1$,
$p_{11} (r(y)):=p_1 (r(y))$ and $p_{21} (r(y)):=p_2 (r(y))$. One has
$p_{11} (r(y))\to -\varphi_+$ and $p_{21} (r(y))\to -\varphi_+$ as $y\to \infty$;
moreover, $\Delta_1(r(y))\to 0$ as $y\to \infty$ and
${p'}_{11},{p'}_{21},{\Delta'}_1 \in L^1 [r_0,\infty)$, with $r_{+}<r_0<\infty$. Then the hypotheses of
theorem 1 p. 185 in \cite{hs-ac} are satisfied, and then the spectrum of $\hat h_{k,j}$
is absolutely continuous in ${\mathbb R}-\{\varphi_+\}$. Cf. also Remark (1) p. 211 in \cite{hs-ac}.
\end{proof}
The analysis of the point $\{\varphi_+\}$ is more involved than in the non-extremal case
and is deferred to further studies.\\
The above result allow us to conclude also that the following result holds for the essential spectrum 
of the radial Hamiltonian:
\begin{corollary}
$\sigma_{e} (\bar{\hat h}_{k,j}) = {\mathbb R}$ both for the non-extremal case and for the extremal one. 
\end{corollary}
\begin{proof}
In the non-extremal case, the result follows from Theorem 5. In the extremal one, is a consequence 
of Lemma 4. 
\end{proof}
Both these results can also be achieved by using Theorem 16.6 in \cite{weidmann}. 
Indeed, let us consider the Hamiltonian (\ref{hamilton-tortoise}) and, in order to match Weidmann's 
conditions, let us introduce the unitary operator $Z:=\left( \begin{array}{cc}
0 & 1\\
1 & 0
\end{array} \right)$. Then let us consider $Z \hat h_{hor} Z^{\ast}$, whose formal expression is 
\eqn
Z h_{hor} Z^{\ast}=\left( \begin{array}{cc}
 0 & \partial_y\\
-\partial_y & 0
\end{array} \right)+\bar{P}(y),
\feqn 
and define  $\bar{P}_0:=\lim_{y\to \infty} \bar{P}(y) =\left( \begin{array}{cc}
\varphi_+ & 0\\
0 & \varphi_+ 
\end{array} 
\right).$ According to Theorem 16.6 in \cite{weidmann}, if 
\eqn
\lim_{y\to \infty} \frac{1}{y} \int_{c}^{y} dt |\bar{P} (t) - \bar{P}_0|=0,
\label{weid-16.6}
\feqn
where $|\cdot|$ stays for any norm for matrices in ${\mathbb C}^{2\times 2}$, then 
${\mathbb R}-\{\varphi_+\}\subset \sigma_e (\hat h_{hor})$, which implies $\sigma_e (\hat h_{hor})={\mathbb R}$. 
By using e.g. the Euclidean norm, 
it is easy to show that in the non-extremal case it holds $ |\int_{c}^{\infty} dt |\bar{P} (t) - \bar{P}_0||<\infty$, 
and then (\ref{weid-16.6}) is implemented. In the extremal case, the integral is divergent as $y\to \infty$ but De l'Hospital's rule 
allows still to conclude that (\ref{weid-16.6}) is implemented. Cf. also 
\cite{belmart} for the standard Kerr-Newman case and \cite{belgcaccia} for the 
Reissner-Nordstr\"{o}m-AdS case.

\subsection{Absence of time-periodic normalizable solutions}

For the non extremal case, the spectral analysis carried out in the previous subsections allows
to conclude that on the given black hole background the Dirac equation does not admit any
normalizable  time-periodic solution. Cf. \cite{finster-axi} and \cite{yamada} for the
Kerr-Newman case.
Both the hypothesis and the proof of such a theorem
are completely analogous to the ones relative to the non-extremal Kerr-Newman black hole case
appearing in \cite{yamada}, theorem IV.5, with simple and obvious replacements.
We can limit ourselves to observe that, given the Dirac equation in its Hamiltonian form
(\ref{hamilton-dirac}), it is possible to obtain solutions which are both normalizable and
time-periodic if and only if the point spectrum of the Hamiltonian $\bar{\hat H}$ is non-empty. 
For the sake of completeness we show that the aforementioned remark holds true. 
We are indebted to Franco Gallone (Universit\`a di Milano) for providing us a rigorous proof of it. 
\begin{remark}
Let $\hat G$ be a selfadjoint operator in a Hilbert space ${\mathcal H}$ and let $P$ the projection 
valued measure of $\hat G$; let $\hat U$ be the 1-parameter strongly continuous group generated by $\hat G$. 
There exist $T\in {\mathbb R}-\{0\}$ and $\psi\in {\mathcal H}-\{0\}$ such that  
$\hat U_{t+T} \psi = \hat U_{t} \psi$ for each $t\in {\mathbb R}$ iff the point spectrum of $\hat G$ 
is non-empty.
\end{remark} 
\begin{proof}
For $T\in {\mathbb R}-\{0\}$ and $\psi\in {\mathcal H}$,  
$\hat U_{t+T} \psi = \hat U_{t} \psi$ is equivalent to $\hat U_{T} \psi =  \psi$, which amounts to 
the condition 
\eqn
\int_{\mathbb R} |{\mathrm{e}}^{i\lambda T} -1|^2 d\mu_{\psi}^{(P)} (\lambda) =0,
\feqn
(where for each $E$ contained in the Borel $\sigma$-algebra ${\mathcal B}$ in ${\mathbb R}$, 
the measure $\mu_{\psi}^{(P)}$ 
is defined by $\mu_{\psi}^{(P)}(E):=||P(E)\psi||^2$).\\
The above condition is implemented iff ${\mathrm{e}}^{i\lambda T} =1$ $\mu_{\psi}$-a.e., which amounts to 
$\mu_{\psi}^{(P)} (\{ \frac{2\pi n}{T} \}_{n\in {\mathbb Z}})=||\psi||^2$, i.e. 
$P(\{ \frac{2\pi n}{T} \}_{n\in {\mathbb Z}})\psi=\psi$. Then there exists a $\psi\ne 0$ such that 
$\hat U_{t+T} \psi = \hat U_{t} \psi$ for every $t\in {\mathbb R}$ iff 
$P(\{ \frac{2\pi n}{T} \}_{n\in {\mathbb Z}})\ne 0$, and the latter condition holds 
iff there exists $n_0\in {\mathbb Z}$ such that 
$P(\{ \frac{2\pi n_0}{T} \})\ne 0$, i.e. such that $\frac{2\pi n_0}{T} \in \sigma_p (\hat G)$. 
\end{proof}
In the case at hand, $\hat G$ is the Hamiltonian operator $\bar{\hat H}$ and $t$ is the time. 
For a non-extremal Kerr-Newman-AdS black hole, we have that the point spectrum of $\bar{\hat H}$ is empty 
(cf. also Lemma 5.3 in \cite{baticschmid} and Proposition 7.1 in 
\cite{hafner} for then Kerr-Newman case)  
and then no time-periodic and normalizable solution of the Dirac equation is allowed. 
{F}rom a physical point of view, this fact means that no quantum mechanical
solution equivalent to a classical closed orbit exists. Cf. \cite{finster-rn,finster-axi}.

\section{Conclusions}

We have considered the Dirac equation on the universal covering of a Kerr-Newman-AdS 
black hole background. The presence of a magnetic charge has been allowed, and the Hamiltonian form
of the Dirac equation has been obtained. Then, we have shown that Theorem 1 holds true, and then 
we have studied essential selfadjointness properties of the Hamiltonian $\hat H$ on 
$C_0^{\infty} ((r_+,\infty) \times S^2)^4$ through the equivalent analysis on $\hat H_0$. 
Variable separation has been performed and 
we have shown that in presence of a magnetic charge the Dirac quantization condition
$\frac{q_m e}{\Xi}\in {\mathbb Z}$ is necessary and sufficient for ensuring essential selfadjointness  
of ${\mathbb U}_k$ on $C_0^{\infty} (0,\pi)^2$ for any $k\in {\mathbb Z}+\frac{1}{2}$. 
Moreover, $\mu l \geq \frac{1}{2}$ has to be implemented in order to obtain essential
selfadjointness of the radial Hamiltonian $\hat h_{k,j}$ on $C_0^{\infty} (r_+,\infty)^2$. 
This is also the condition
one finds on an Anti-de Sitter background.\\
Furthermore, we have taken into account some spectral properties of the Hamiltonian 
for $\mu l \geq \frac{1}{2}$. Qualitative
spectral analysis has allowed us to conclude that $\sigma_e (\hat h_{\infty})=\emptyset$, in 
contrast to what happens in the (asymptotically flat) Kerr-Newman case, 
and also that no time-periodic and normalizable solution of the
Dirac equation is allowed on the given non-extremal black hole background. The latter conclusion is in agreement with
the analogous result for black hole of the Kerr-Newman family and more in general for axi-symmetric
black holes which are asymptotically flat \cite{finster-axi,yamada}. Moreover,  in the case $\mu l < \frac{1}{2}$ 
this holds true also for selfadjoint extensions $\hat T_H$ obtained by imposing local boundary conditions at 
infinity. Cf. the comment following Lemma 2.\\
The implementation of a second-quantization formalism and the analysis of the mechanism
allowing both the discharge and the loss of angular momentum by the black hole by means
of quantum effects \cite{damo,belmart} deserve future investigations.

\section*{{\bf Acknowledgments}}

We are indebted to Franco Gallone (Universit\`a degli Studi di Milano) for several discussions and
fruitful comments and  for a rigorous proof of Remark 1. 
We thanks also Aldo Treves, Andrea Posilicano and Stefano Pigola (Universit\`a dell'Insubria)
for discussions and remarks.

\appendix

\section{Essential selfadjointness of $\hat{\mathbb U}$}
\label{essautoU}

For the sake of completeness, we discuss in detail the essential selfadjointness conditions for
the angular momentum operator ${\mathbb U}$. Let us introduce $d:= \frac{q_m e}{\Xi}$ and also the
sets
\eqn
I_0:= (-1+d,d),
\feqn
which is such that for $n\in I_0$ the limit circle case \cite{weidmann} is implemented at $\theta=0$, and
\eqn
I_{\pi}:= (-1-d,-d),
\feqn
which is such that for $n\in I_{\pi}$ the limit circle case is implemented at $\theta=\pi$.
Essential selfadjointness for a generic $n\in {\mathbb R}$ is implemented for
\eqn
n\in {\mathbb R}-(I_0\cup I_{\pi}).
\label{esscondU}
\feqn
Condition (\ref{esscondU}) for $|d|\leq \frac{1}{2}$ amounts to
\eqn
n\in (-\infty,-1-|d|]\cup [|d|,\infty).
\label{condmin}
\feqn
If $|d|>\frac{1}{2}$, condition (\ref{esscondU}) amounts to
\eqn
n\in (-\infty,-1-|d|]\cup [-|d|,-1+|d|]\cup [|d|,\infty).
\label{condmax}
\feqn
Now we make use of the fact that $n\in {\mathbb Z}$ in our case. For $d\in {\mathbb Z}$,
essential selfadjointness is implemented without restrictions. If $|d|:=[|d|]+\zeta$,
where $\zeta\in (0,1)$, then essential selfadjointness is implemented for
$n\in {\mathbb Z}-\{-1-[|d|],[|d|] \}$, which  requires further analysis. 
We can note that choosing $d\in {\mathbb Z}$ amounts to the Dirac quantization condition in the 
stronger form which was considered as mandatory by Schwinger in his model for a relativistic 
quantum field theory for fermionic matter in presence of a magnetic charge \cite{schwinger} and 
in which only integer values are allowed instead of semi-integer values. See also \cite{physrep}. 
This quantization rule is confirmed also by a construction \`a la Wu and Yang \cite{wu}. 
See e.g. \cite{ghosh} (formula (11) therein and its consequences for $a=0$ on the charge 
quantization, which corresponds to our case).\\
For the sake of completeness, we point out that 
semi-integer values for $d$ would introduce the necessity of implementing 
suitable boundary conditions at $\theta=0$ or at $\theta=\pi$ for ${\mathbb U}_k$ for special values of $k$, 
which are such that the defect indices of ${\mathbb U}_k$ are $(1,1)$. Let us put $|d|=l+\frac{1}{2}$, 
with $l\in {\mathbb N}$.  
Then we find that $n=-l-1$ and $n=l$ do not satisfy the essential selfadjointness conditions 
indicated above. Being $I_0\cap I_{\pi}=\emptyset$ for $|d|=l+\frac{1}{2}$, which means that the limit 
circle case occurs only at one of the extremes of $(0,\pi)$, one finds that 
the partial wave operators ${\mathbb U}_{-l-\frac{1}{2}}$ and ${\mathbb U}_{l+\frac{1}{2}}$ have defect indices $(1,1)$. 
This would introduce an ``asymmetric'' treatment for a couple of partial waves labeled by $k$ with respect to 
all the other ones which do not require boundary conditions. 
This ``asymmetry'' would not be justified by any physical argument, being the singular behavior at $\theta=0$ 
or at $\theta=\pi$ only due to a bad behavior of the chart one is forced to introduce if a single-chart 
description of the 1-form connection $A$ is adopted, as in our case; then the following choice is 
taken into account: a core for the extension is identified with 
$C_0^{\infty} (0,\pi)^2+L \{ w_+\}$ \cite{weidmann}, where $w_{+}$ is such that near the point $\theta_0$ 
($\theta_0=0$ or $\theta_0=\pi$) at 
which the limit circle case occurs one has $w_{+}\sim u_{s}$ 
($u_{s}\sim  (\theta-\theta_0)^s$ being the fundamental solution of the equation ${\mathbb U}_k u -\lambda u=0$ 
near $\theta_0$ with $s=\max (\pm \nu)$  if $\theta_0=0$, cf. (\ref{enne0}), or $s=\max (\pm \rho_0)$ 
if $\theta_0=\pi$, cf. (\ref{emme0})), and $w_{+}\sim 0$ at the other extreme. This is the only choice 
ensuring, together with the condition $\mu l\geq \frac{1}{2}$, that  $\hat H_0$ is essentially selfadjoint 
on $C_0^{\infty} ((r_+,\infty) \times S^2)^4$ (lack of smoothness at $\theta_0$ occurs for any other 
choice).\\ 
It must be pointed out that another choice of the gauge would be possible and that it would 
lead to different conditions. 
At least almost everywhere, gauge equivalence of the given $A$ with 
\eqn
A+b \left( \frac {q_m}{\rho \sqrt {\Delta \theta} \sin \theta} e^1-\frac 
{q_m a}{\rho \sqrt {\Delta_r}}e^0 \right)
\feqn
can be shown for any value of the real constant $b$. Then the 
replacements $q_m \cos \theta \mapsto q_m (\cos \theta -b)$  and $q_e r \mapsto q_e r +b q_m a$ 
occur. 
Essential selfadjointness properties of $\hat h_{k,j}$ remain unaltered for any $b$. 
Instead, if for example one chooses $b=1$, essential selfadjointness properties of 
${\mathbb U}_k$ are affected in an evident manner: this choice, which  
corresponds to the usual choice introduced since the seminal paper by Dirac, 
is such that all problems are shifted to $\theta=\pi$ and  
conditions (\ref{condmin}) and (\ref{condmax}) are replaced by the essential selfadjointness condition
\eqn
n\in (-\infty,-1-2 d]\cup[-2 d,\infty)
\label{condirac}
\feqn
Then, semi-integer values of the magnetic charge ensure essential selfadjointness too, 
in agreement with the general form of the Dirac quantization condition. We are not 
aware of a solution of the dichotomy between the physical situation represented by the 
Dirac string and the Schwinger one but for the explanation given in \cite{barut}: in fact 
for a fixed magnetic field the infinite singularity line embodied by the choice of the Schwinger potential 
is associated with a monopole of double strength with respect to the one which is 
associated with the Dirac semi-infinite singularity line (a double flux is generated by the former 
with respect to the latter). Hence a factor 2 appears. In other terms, the relation 
$q_m^{Schwinger}=2 q_m^{Dirac}$ should occur. 
Then it should be also true that 
the validity of the aforementioned gauge equivalence is only almost everywhere (smooth part) 
and that the singular part cannot be gauge-equivalent due to the different behavior of a 
semi-infinite string with respect to an infinite one. 
This interpretation would be also corroborated by the fact that, as it is evident, the 
corresponding transformation of the Hamiltonian cannot be implemented by means of a unitary transformation 
(a unitary transformation would preserve the essential self-adjointness properties).

\section{Spectrum of the operator $\bar{\hat{\mathbb {U}}}_{k\; \omega}$. Alternative proof.}
\label{spectrumU}

For simplicity, we take into account only the case of vanishing magnetic charge $q_m=0$.
We mean to make use of theorem 3 in \cite{hs-discrete}. The multiplication operator by 
\eqn
R W {\mathbb {V}}_{\omega}  W^{\ast} R^{\ast} =
a \omega \frac{\sin (\theta)}{\sqrt{\Delta_{\theta}}}
\left(
\begin{array}{cc}
0 & 1\\
1 & 0
\end{array}
\right)
\feqn
is a bounded perturbation of the operator $ R W\bar{\hat {\mathbb {U}}}_k  W^{\ast} R^{\ast}$. Then we consider first
the spectrum of the latter operator. In order to apply the theory described in \cite{hs-discrete}
it is useful to rewrite the (unperturbed) differential system 
$R W {\mathbb U}_k  W^{\ast} R^{\ast} \Theta = \lambda \Theta$ 
as follows:
\eqn
\Theta' = -\frac{k \Xi}{\Delta_{\theta} \sin(\theta)}
\left(
\begin{array}{cc}
1 & \frac{1}{k\Xi} [(\mu a \cos (\theta)+\lambda) \sqrt{\Delta_{\theta}} \sin(\theta)]\\
\frac{1}{k\Xi} [(\mu a \cos (\theta)-\lambda) \sqrt{\Delta_{\theta}} \sin(\theta)] & -1
\end{array}
\right) \Theta \Leftrightarrow \Theta' =: {\mathcal U} \Theta,
\label{ang-hs}
\feqn
where the prime stays for the derivative with respect to $\theta$. (Notice that, by comparing (1.4) of
\cite{hs-discrete} with our operator, we have to shift $\lambda \mapsto -\lambda$ in (1.1) of
\cite{hs-discrete}).
According to the theory in \cite{hs-discrete}, we are in the diagonally dominant case. Moreover,
in order to face with our problem, which displays two singular endpoints, we refer to Remark (3),
pp. 119-120, of \cite{hs-discrete}, according to which if discrete spectrum criteria are satisfied
at both endpoints, then the spectrum is discrete. 
Let us introduce two auxiliary selfadjoint operators: 
\eqn
&&D(\hat{\mathbb U}_0)=\{ \Theta\in L^2 ((0,c), \frac{d\theta}{\sqrt{\Delta_{\theta}}})^2;
\Theta\; \hbox{is locally absolutely continuous}; B (\Theta) =0;
\hat{\mathbb U}_0 \Theta\in L^2 ((0,c), \frac{d\theta}{\sqrt{\Delta_{\theta}}})^2 \},\cr
&&\hat{\mathbb U}_0\; \Theta = R W {\mathbb U}_k  W^{\ast} R^{\ast}\; \Theta, \quad \Theta \in D(\hat{\mathbb U}_0);\cr
&&D(\hat{\mathbb U}_{\pi})=\{ \Theta\in L^2 ((c,\pi), \frac{d\theta}{\sqrt{\Delta_{\theta}}})^2;
\Theta\; \hbox{ is locally absolutely continuous}; B (\Theta) =0;
\hat{\mathbb U}_{\pi} \Theta\in L^2 ((c,\pi), \frac{d\theta}{\sqrt{\Delta_{\theta}}})^2 \},\cr 
&&\hat{\mathbb U}_{\pi}\; \Theta = R W {\mathbb U}_k  W^{\ast} R^{\ast}\; \Theta, \quad \Theta \in D(\hat{\mathbb U}_{\pi}).
\feqn 
$0<c<\pi$ is an arbitrary (regular) point at which the boundary condition
$B (\Theta) = \sin (\beta) \Theta_1 (c) + \cos (\beta) \Theta_2 (c)=0$, with $\beta\in [0,\pi)$,
is imposed for both.\\  
According to the notation in \cite{hs-discrete}, we have $p(\theta) = -\frac{k\Xi}{\Delta_{\theta} \sin(\theta)}$,
$p_1 (\theta)=p_2 (\theta)= -\mu a \frac{\cos (\theta)}{\sqrt{\Delta_{\theta}}}\in L^1 (0,\pi)$, and moreover
$\alpha_1 (\theta) = \alpha_2 (\theta) = \frac{1}{\sqrt{\Delta_{\theta}}} \in L^1 (0,\pi)$. Then
$\frac{\alpha_k (\theta)}{p(\theta)}=-\sqrt{\Delta_{\theta}} \sin(\theta) \frac{1}{k\Xi} := \alpha_{k3}$,
with $k=1,2$; furthermore, $a_1=0=a_2$ and then $q_k (\theta)=p_k (\theta)$, and
$\frac{q_k (\theta)}{p(\theta)}=q_{k3} (\theta)$. According to the definitions in \cite{hs-discrete} (p.102),
$\alpha_1 (\theta) = \alpha_2 (\theta)$ and $p_k (\theta)$ are short range. As a consequence, we get
$\Gamma_1 = \Gamma_2 =0$ and $q_{k1}=0=q_{k2}$ (cf. \cite{hs-discrete} p. 112). Then in the notation of
\cite{hs-discrete}, p. 112, we obtain
\eqn
D_1 (\theta) =
\left(
\begin{array}{cc}
1  &  0\\
0  & -1
\end{array}
\right),
\feqn
and
\eqn
D_3 (\theta) =
\left(
\begin{array}{cc}
0                            &  q_{23}-\lambda \alpha_{23}\\
-q_{13}+\lambda \alpha_{13}  & 0
\end{array}
\right),
\feqn
whereas the matrix $D_2 (\theta)$ in our case is zero. Moreover we obtain
\eqn
\mu_0 (\theta, \lambda)=1
\feqn
and, as a consequence, the matrix ${\mathcal S}$ defined in \cite{hs-discrete}, p. 113, is the
identity matrix in our case. One also obtains
\begin{eqnarray}
E(\theta,\lambda) &=& \exp \left( \int_c^{\theta} dt p(t) \right) \cr
& = & \exp \left[- \frac{k}{2} \left( \frac{a}{l} \log \frac{l+a \cos (t)}{l-a \cos (t)}
-\log \frac{1+\cos (t)}{1-\cos (t)} \right) \right]\bigg|_{c}^{\theta},
\label{etheta}
\end{eqnarray}
and can define the matrix
\eqn
\Omega (\theta, \lambda)=
\left(
\begin{array}{cc}
E(\theta,\lambda)  &  0\\
0  & \frac{1}{E(\theta,\lambda)}
\end{array}
\right).
\feqn
The matrix $B(\theta)$ defined in \cite{hs-discrete}, p. 113, is zero. The matrix
$C(\theta)$ in our case coincides with $p(\theta) D_3 (\theta)$ and is absolutely integrable.
The matrix $G$ defined at p. 113 of \cite{hs-discrete} is identically zero in our case.
It is then easy to show that the hypotheses of theorem 3, p.113, of \cite{hs-discrete} are
satisfied. The limit point case criterion given in theorem 3 of \cite{hs-discrete}
is easily verified both at $\theta=0$ and at $\theta=\pi$ due to (\ref{etheta}),
confirming the analysis carried out in the previous sections. Moreover, the
criteria $\int_c^0 dt |p(t)|=\infty$ and $\int_c^{\pi} dt |p(t)|=\infty$ are both
satisfied, and then both $\hat{\mathbb U}_0$
and $\hat{\mathbb U}_{\pi}$ have discrete spectra. As a consequence, $\bar{\hat{\mathbb U}}_k$ has discrete
spectrum. As a consequence of theorem 10.8 in \cite{weidmann}, we can also conclude that
the spectrum is simple.\\
The bounded perturbation $W {\mathbb {V}}_{\omega} W^{\ast}$ affects the previous analysis in the following sense: the
discrete eigenvalues $\lambda$ get an analytic dependence on $\omega$. Indeed, $W {\mathbb {V}}_{\omega} W^{\ast}$
is infinitesimally small with respect to $W \bar{\hat {\mathbb {U}}}_k W^{\ast}$ \cite{RSII} and then
$W {\mathbb {U}}_{k\; \omega} W^{\ast}$ defines an analytical family of type (A) according to Kato's definition
\cite{kato,RSIV}. See in particular \cite{RSIV}, p.16.\\

\section{Discrete spectrum of $\hat h_{\infty}$. Alternative proof}
\label{spectrumhinf}

We refer to the results contained in \cite{hs-discrete}.
It is useful to rewrite (\ref{radialequation}) in the following form:
\eqn
X' =
\left(
\begin{array}{cc}
p(r) & -\omega \alpha (r) + p_2 (r)\\
\omega \alpha (r) - p_1 (r) & -p(r)
\end{array}
\right) X,
\label{eqradial}
\feqn
where $p(r):=\frac{\lambda_{k;j}}{\sqrt{\Delta_r}}$, $\alpha (r):=\frac{r^2+a^2}{\Delta_r}$,
$p_1 (r):=\frac{1}{\Delta_r} ( a \Xi k + e q_e r +\mu r \sqrt{\Delta_r} )$ and
$p_2 (r):=\frac{1}{\Delta_r} ( a \Xi k + e q_e r -\mu r \sqrt{\Delta_r} )$.\\

We study the spectral properties of $\hat h_{\infty}$.
In order to apply the theory exposed in \cite{hs-discrete}, theorem 1, pp. 102-103,
we put $p_{11}(r):=\mu r \frac{1}{\sqrt{\Delta_r}}$,
which is positive as far as $\mu>0$ in the physical case, and $p_{21}(r):=-\mu r \frac{1}{\sqrt{\Delta_r}}$. As
a consequence, being by assumption $p_1 (r)=p_{11}(r)+p_{12} (r)$ and analogously
$p_2 (r)=p_{21}(r)+p_{22} (r)$, one finds $p_{12} (r)=\frac{1}{\Delta_r} ( a \Xi k + e q_e r)$
and $p_{22} (r)=\frac{1}{\Delta_r} ( a \Xi k + e q_e r)$; furthermore, one obtains
$Q(r): = \sqrt{-p_{11}(r) p_{21}(r)}=p_{11}(r)$. Moreover, one gets
$\alpha_1 (r)=\alpha_2 (r)=\alpha (r)$, and then
$\frac{-\alpha_k (r)}{p_{k1} (r)}=r_{k3} (r)$, with $k=1,2$, where
$Q(r) r_{k3} (r) \in L^1 [r_0,\infty)$ (note that there is a misprint in \cite{hs-discrete}, p. 102,
regarding this condition: $\frac{-\alpha_k (r)}{p_{k2} (r)}$ is indicated in place
of $\frac{-\alpha_k (r)}{p_{k1} (r)}$). The latter property means that
$Q(r) r_{k3} (r)$ is short range (cf. \cite{hs-discrete}, p. 102) (note that $Q(r) \sim \frac{\mu l}{r}$
and $\frac{-\alpha_k (r)}{p_{k1} (r)}\sim (-)^{k} \frac{l}{\mu} \frac{1}{r}$ as $r\to \infty$).
As to the ratio $\frac{-p_{k2} (r)}{p_{k1} (r)}$, we find that both the terms are short range, and then
we can put $\frac{-p_{k2} (r)}{p_{k1} (r)}=s_{k3} (r)$. In our case we get
$r_{k1} (r)=0=r_{k2} (r)$ and $s_{k1} (r)=0=s_{k2} (r)$. The
function $\bar{\Delta}:=(p^{-1}_{11}(r) p'_{11}(r)-p^{-1}_{21}(r) p'_{21}(r)-p(r)) \frac{1}{Q(r)}$
(cf. \cite{hs-discrete}, p. 103) is such that $\bar{\Delta}=-\frac{\lambda_{k;j}}{\mu r}=\bar{\Delta}_3$,
where $\bar{\Delta}_3 Q(r)\in L^1 [r_0,\infty)$ and then is short range.\\
The function $\mu_0 (r,\omega)$ in our case
is $1$ and then the function $E (r,\omega)$ defined at p. 103 of \cite{hs-discrete} is
\eqn
E (r,\omega) = \exp \left( \int_{r_0}^{r} dt Q(t) \right)
\feqn
The condition (1.5) at p. 103 of of \cite{hs-discrete}, which is necessary and sufficient in order
to get the limit point case at infinity, in our case becomes:
\eqn
\int_{r_0}^{\infty} dt (E^2 (r,0)+E^{-2} (r,0)) (2 \alpha (r))=\infty;
\feqn
being $E^{\pm 2} (r,0)\sim r^{\pm 2 \mu l}$ and $\alpha (r)\sim \frac{l^2}{r^2}$ for
$r\to \infty$, one obtains the same conditions we found in our study of the essential selfadjointness
of the radial Hamiltonian. Moreover, when the limit point case is implemented, the condition
\eqn
\int_{r_0}^{\infty} dt Q(t)=\infty
\feqn
is sufficient in order to get a discrete spectrum, and this condition is easily verified in our case.
Then the spectrum of $\hat h_{\infty}$ is discrete. Note that, if the limit circle case occurs
at infinity, the spectrum is still discrete (cf. theorem 7.11 in \cite{weidmann}).

\section*{{\bf Note added}}

After this work was completed, we were made aware of Ref. \cite{winklyamada}. Therein a result 
about essential selfadjointness of $\hat H$ which is analogous to theorem 1 is stated for the Kerr-Newman case 
(cf. theorem 2.1 therein).

\end{document}